\documentclass[journal,10pt]{IEEEtran} 
\makeatletter
\def\endthebibliography{%
	\def\@noitemerr{\@latex@warning{Empty `thebibliography' environment}}%
	\endlist
}
\makeatother

\usepackage{cite}

\usepackage[pdftex]{graphicx}

\usepackage[caption=false,font=footnotesize]{subfig}

\usepackage{amsmath}
\usepackage{mathtools, cuted}
\usepackage{filecontents}
\usepackage{amssymb}
\usepackage{color}
\usepackage{xcolor}

\usepackage{epsfig}

\usepackage{url}

\usepackage{algpseudocode}
\usepackage{algorithm, tabularx}

\usepackage{lettrine}

\usepackage{lipsum}

\usepackage{siunitx}

\usepackage{soul,color}

\usepackage{mathrsfs}

\usepackage{array}

\usepackage[inline]{enumitem}

\usepackage{breqn}

\usepackage{multirow}
\usepackage{array}
\newcolumntype{L}[1]{>{\raggedright\let\newline\\\arraybackslash\hspace{0pt}}m{#1}}
\newcolumntype{C}[1]{>{\centering\let\newline\\\arraybackslash\hspace{0pt}}m{#1}}
\newcolumntype{R}[1]{>{\raggedleft\let\newline\\\arraybackslash\hspace{0pt}}m{#1}}

\usepackage{amsthm}
\newtheorem{theorem}{Theorem}

\newlength{\maxwidth}
\newcommand{\algalign}[2]
{\makebox[\maxwidth][r]{$#1{}$}${}#2$}

\makeatletter
\newcommand{\multiline}[1]{%
	\begin{tabularx}{\dimexpr\linewidth-\ALG@thistlm}[t]{@{}X@{}}
		#1
	\end{tabularx}
}
\makeatother

\theoremstyle{remark}

\newtheorem{remark}{Remark}

\algdef{SE}[SUBALG]{Indent}{EndIndent}{}{\algorithmicend\ }%
\algtext*{Indent}
\algtext*{EndIndent}

\begin{document}
	
	\title{Ampli-Flection for 6G: Active-RIS-Aided\\Aerial Backhaul with Full 3D Coverage}

	\author{Hong-Bae Jeon,~\IEEEmembership{Member,~IEEE}, and Chan-Byoung Chae,~\IEEEmembership{Fellow,~IEEE}
		
		\thanks{This work was supported by Hankuk University of Foreign Studies Research Fund of 2026. This article was presented in part at the IEEE ICTC 2022, Jeju Island, South Korea, October 19-21, 2022~\cite{HBICTC}. \textit{(Corresponding Author: Chan-Byoung Chae.)}}
		\thanks{H.-B. Jeon is with the Department of Information Communications Engineering, Hankuk University of Foreign Studies, Yong-in, 17035, Korea (e-mail: hongbae08@hufs.ac.kr).}
		\thanks{C.-B. Chae is with the School of Integrated Technology, Yonsei University, Seoul 03722, Korea (e-mail: cbchae@yonsei.ac.kr).}
	}
	

	\maketitle

	\begin{abstract}
	In this paper, we propose a novel aerial backhaul architecture that employs an aerial active reconfigurable intelligent surface (RIS) to achieve energy-efficient, {full 3D coverage including UAV-BSs and ground users in 6G wireless networks}. Unlike prior aerial-RIS approaches limited to {2D coverage with only servicing ground users} or passive operation, the proposed design integrates an active-RIS onto a high-altitude aerial platform, enabling reliable line-of-sight links and overcoming multiplicative fading through amplification. In a scenario with UAV-BSs deployed to handle sudden traffic surges in urban areas, the aerial-active-RIS both reflects and amplifies backhaul signals to overcome blockage. We jointly optimize the aerial platform placement, array partitioning, and RIS phase configuration to maximize UAV-BS energy-efficiency. Simulation results confirm that the proposed method significantly outperforms benchmarks, demonstrating its strong potential to deliver resilient backhaul connectivity with comprehensive 3D coverage in 6G networks.

    	\end{abstract}


	\IEEEpeerreviewmaketitle
	
	\section{Introduction}
\lettrine{O}{F LATE}, the rapid expansion of data-driven applications with the push for extended wireless coverage in beyond fifth-generation (B5G) and sixth-generation (6G) networks have driven an unprecedented demand for providing ultra-reliable and high-capacity communications anywhere~\cite{smidaFD, bjorn6G, yhFD}. In response, researchers have been exploring the deployment of aerial base stations, realized through unmanned-aerial-vehicle-base-stations (UAV-BSs), for rapid and flexible coverage extensions~\cite{ntn6g, mozatut}. These UAV-BSs present a compelling solution for extending network coverage to areas where conventional base station infrastructure is inadequate, offering rapid deployment capabilities in emergency or high-traffic scenarios. Their high-altitude operation further enhances the likelihood of rich line-of-sight (LoS) links~\cite{a2gglobecom,Noh}, thereby improving communication reliability and efficiency. To enable efficient UAV-BS deployment, prior studies have investigated deployment strategy design~\cite{uavdep1, uavdep2}, trajectory optimization~\cite{hjgen, traj3d}, and power control schemes~\cite{hjcoop, uavmec}. {One of the critical challenges} that remains for deploying UAV-BSs is the backhaul connectivity, particularly in dense urban areas where obstacles obstruct direct backhaul links from terrestrial sources and generates non-LoS (NLoS) components, which severely deteriorates the system’s energy-efficiency, posing a critical challenge for 6G wireless networks operating with higher spectrum~\cite{HBFSO, yhFD22, dsRIS}.

Therefore, the reconfigurable intelligent surfaces (RIS) have been explored as a means to enhance wireless coverage and mitigate blockages~\cite{RIST, risspm, add3,LingRIS}. RIS is an artificial metasurface composed of passive reflecting elements capable of adjusting both the amplitude and phase of an incident signal~\cite{add1, add2}. Furthermore, due to its passive array architecture, the RIS exhibits high channel capacity with low power consumption, which leads to the energy-efficiency compared to conventional decode-and-forward or amplify-and-forward (DF/AF) relays in several scenarios~\cite{DF, vsrelay, Doh}. However, conventional RIS has passive reflection elements, which highly suffer from severe ``multiplicative fading,'' where the path loss is given by the multiplication of Tx-to-RIS and RIS-to-Rx links, limiting their effectiveness in practical scenarios~\cite{aris1, aris4}. {These disadvantages become} more severe when the RIS is deployed in terrestrial infrastructure on dense urban areas.

{To overcome the fundamental physical limitation imposed by the multiplicative fading effect of cascaded channels in RIS-assisted systems, the concept of active-RIS has recently been proposed~\cite{aristut, aris1}. Unlike conventional passive-RISs, which only reflect incident signals without amplification, active-RISs incorporate reflection-type amplifiers within each element, thereby enabling signal amplification at the cost of additional power consumption~\cite{alexRIS}. This capability fundamentally changes the power-rate trade-off and opens up new possibilities for performance enhancement in future 6G wireless networks.}

{Early studies, such as~\cite{aris1}, focused on characterizing the performance gains of active-RIS over passive-RIS through sum-rate maximization and comparative analyses, demonstrating that active-RIS can effectively mitigate the double-fading loss inherent in passive architectures. These works established the potential of active-RIS as a
promising technology for next-generation systems. Building upon this foundation, subsequent research addressed practical implementation challenges. In particular,~\cite{aris2} investigated a sub-connected active-RIS architecture and proposed a joint beamforming design that balances hardware complexity and performance, revealing important trade-offs between achievable gains and circuit-level constraints. From a performance perspective,~\cite{aris4, aris5} further examined the SNR advantages of active-RIS under a unified power budget,
showing that, when properly optimized, active-RIS can significantly outperform passive-RIS counterparts. More recently after performance characterization and architectural investigations, attention has shifted toward robustness and practical deployability. To bridge the gap between theoretical analysis and real-world operation,~\cite{aris7, aris8} studied active-RIS transmission designs under partial channel state information (CSI). Specifically,~\cite{aris7} aimed at maximizing the average sum-rate, while~\cite{aris8} focused on minimizing the average total transmit power subject to rate and outage probability constraints. These works highlight the importance of accounting for CSI uncertainty when designing active-RIS-assisted systems.}

Most prior work, however, has primarily focused on deploying the active-RIS in terrestrial environments, such as those involving buildings and walls. These environments impose several limitations on communication performance. In dense urban areas with many buildings, reliable transmission often depends on multiple reflections, which requires deploying a large number of RIS units to alleviate severe signal attenuation. Moreover, a terrestrial RIS can only reflect signals from the source to destinations located on the same side, thereby restricting the angular range of reflection and preventing isotropic coverage over $0^{\circ}\sim360^{\circ}$ arrival angles. 
{Several studies have investigated the deployment of RIS on aerial platforms (e.g.,~\cite{aerial3d, UPARIS, aerialmm}); however, most of these works focus on conventional passive RIS architectures and consider only 2D coverage, or assume fixed transmit power at the source, thereby largely overlooking energy-efficiency considerations. Although a few recent studies have explored active-RIS-assisted aerial systems~\cite{arisb1, arisb2, arisb3}, they still restrict their scope to 2D coverage scenarios\footnote{{In this paper, we define ``2D coverage'' as the servicing of only ground users, whereas ``3D coverage'' refers to the service of both aerial and ground users, which constitutes a geometrical extension of 2D coverage. It is emphasized that the term ``3D coverage'' in this work refers to geometry-aware coverage over a three-dimensional spatial region, rather than full-dimensional beamforming enabled by planar arrays.}}, which limits their applicability in dense urban 6G environments with heterogeneous mobility and altitude-dependent service demands.}

{In contrast, this paper proposes a novel aerial active-RIS architecture that explicitly distinguishes the geometric benefits of aerial deployment from the amplification-enabled gains introduced by the active-RIS itself. While the elevated platform inherently improves LoS availability and alleviates blockage through favorable 3D geometry, the active-RIS provides an additional and fundamentally different advantage by amplifying the reflected signals at the element level. This amplification capability mitigates the multiplicative path-loss of cascaded source-RIS-destination links, which is an important issue due to the long link distance of aerial platform and cannot be achieved by aerial placement or passive RIS alone under the same power budget~\cite{arisleo}. As a result, the proposed architecture enables energy-efficient 3D coverage and, for the first time, facilitates an aerial-active-RIS-assisted backhaul network capable of reliably supporting UAV-BSs and heterogeneous users in future 6G wireless networks.}


The key contributions are summarized as follows:
\begin{enumerate}
\item We propose a novel aerial-active-RIS architecture aimed at enhancing 3D backhaul connectivity to UAV-BSs in urban environments characterized by severe blockages. The proposed scheme deploys an RIS on a high-altitude platform, where each RIS element is equipped with an active amplifier, ensuring both a rich LoS component and robustness against multiplicative fading, respectively. To maximize the received SNR and minimize the transmit power, we show that the maximum-ratio transmission (MRT) strategy achieves the purpose.
\item We verify that equal amplification gain across active-RIS elements is a feasible approach for optimizing the energy-efficiency by minimizing the total consumed power of the system. Under this assumption, an optimization framework is developed to determine the placement, array-partitioning strategy, phase control and optimal amplification gain of the aerial-active-RIS with the objective of maximizing energy-efficiency. Specifically, a nonzero closed-form value of the source transmit power for each UAV-BS is derived, and a minimization problem is formulated by adjusting the numerator and denominator of the value. {The problem is efficiently solved using the global criterion method, which selects a minimal-total-distance operating point corresponding to a Pareto-efficient solution, and considering the partition of full RIS array, respectively.}
\item The proposed approach is evaluated numerically in a realistic urban outdoor environment with $10^3$ randomly distributed ground users and corresponding UAV-BSs. Extensive numerical results demonstrate that the proposed aerial-active-RIS scheme significantly outperforms conventional benchmarks including {the aerial-AF-relay} and the aerial-passive-RIS schemes and aerial-active-RIS with randomly determined amplification gain in terms of energy-efficiency.
\end{enumerate}

\begin{figure*}[t]
	\begin{center}
		\includegraphics[width=1.98\columnwidth,keepaspectratio]%
		{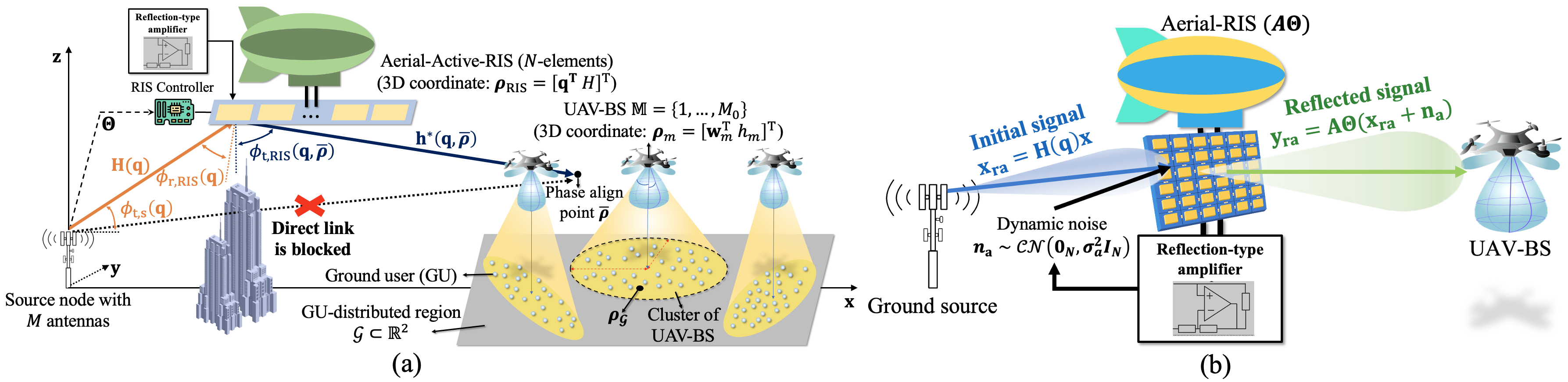}
		\caption{(a) UAV-BS access network supported by an aerial-active-RIS backhaul and (b) signal model illustration of the aerial-active-RIS.}
		\label{fig_overview}
	\end{center}
\end{figure*}	
\section{System Model}\label{4.1}
	\subsection{Aerial Backhaul with Conventional Passive-RIS}\label{4.1.1}
We consider an urban area $\mathcal{G}$ with origin $\pmb{\rho}_{\mathcal{G}}$ on xy-plane, as depicted in Fig.~\ref{fig_overview}(a). We assume that $\mathcal{G}$ contains $N_0$ ground users equipped with omnidirectional antenna, and served by multiple stationary UAV-BSs with a directional antenna, represented as $\mathbb{M}=\left\{1,\cdots, M_0 \right\}$. The directional antenna of each UAV-BS has different azimuth and elevation half-power beamwidths (HPBW)~\cite{beamwidth}. The 3D coordinates of a UAV-BS $m$ are expressed as $\pmb{\rho}_m\triangleq\left[\mathbf{w}_m^{\mathrm{T}}~h_m\right]^{\mathrm{T}}$, with its 2D coordinate $\mathbf{w}_m$ and altitude $h_m$.

To alleviate interference among adjacent cells, we assume that each UAV-BS serves a distinct, non-overlapping subset of users. In this context, the Ellipse Clustering algorithm~\cite{Noh} is employed to allocate ground users and determine $M_0$, while significantly reducing the transmit power through optimized 3D deployment. Under this setup, the throughput $C_m$ of UAV-BS $m$ is expressed as
\begin{equation}
\label{throughput}
C_m=\sum_{n\in\mathbb{U}_m} \frac{B_{\mathrm{f}}}{\left|\mathbb{U}_m\right|} \log_2 \left(1+ \frac{P_{\mathrm{f},n}}{\sigma_{\mathrm{f}}^2}\right),
\end{equation}
where $\mathbb{U}_m$ represents the user set served by UAV-BS $m$, $B_{\mathrm{f}}$ denotes the fronthaul bandwidth, which is evenly divided among the $\left|\mathbb{U}_m\right|$ users in $\mathbb{U}_m$, $P_{\mathrm{f},n}$ is the received power at user $n\in\mathbb{U}_m$, and $\sigma_{\mathrm{f},m}^2$ 
represents the noise power of the fronthaul link for UAV-BS $m$. 
{By adopting the frequency-division-multiple-access (FDMA) within each cell and non-overlapping UAV-BS coverage regions, both intra- and inter-cell interference are eliminated, and the achievable rate in~\eqref{throughput} directly follows.}

{For an aerial-active-RIS which maintains an altitude $H$, and designating the first element as the reference, the 3D coordinates of the aerial-RIS are expressed as $\pmb{\rho}_{\mathrm{RIS}}\triangleq\left[\mathbf{q}^{\mathrm{T}}~H\right]^{\mathrm{T}}$. We consider a source at origin, which is equipped with a uniform linear array (ULA) comprising $M$ antennas with inter-element spacing $d_{\mathrm{s}}$ and an antenna gain of $G_{\mathrm{s}}$. Note that although the source and the aerial-active-RIS employ ULAs, the considered coverage is inherently ``three-dimensional'', as the RIS and UAV-BS location, propagation distance, and elevation angle are explicitly characterized in a 3D coordinate system~\cite{HBRIS, ula1, ula2}. In particular, $H$ and $\{h_m\}$ induce elevation-dependent path loss and phase variations, which are fully captured in the channel model.}

The distance between the source and the center of the coverage area is denoted by $d_{\mathcal{G}}\triangleq||\pmb{\rho}_{\mathcal{G}}||_2$, and is assumed to be sufficiently large and the direct transmission path from the source to the UAV-BS is blocked. {The aerial-RIS is modeled as a ULA comprising $N$ reflecting elements with inter-element spacing $d_{\mathrm{RIS}}$, and carrier wavelength $\lambda$. We assumed an uncoupled model for the RIS elements under this ULA-spacing configuration~\cite{UAVRIS, ularis1, ularis2}.} Without loss of generality, the RIS is assumed to be aligned parallel to the x-axis.

Due to the elevated altitude of the aerial-RIS, the backhaul link is assumed to be dominated by {an LoS} component. Considering that $d_{\mathrm{RIS}}$ is significantly smaller than both $H$ and $d_{\mathcal{G}}$, 
the backhaul link is approximated using a uniform plane-wave model, implying that the path loss is assumed to be identical across all RIS-element pairs. Accordingly, the path loss associated with the source-to-RIS link, $\beta_{\mathrm{s}}\left(\mathbf{q}\right)$, and the RIS-to-UAV-BS link, $\beta\left(\mathbf{q},  \pmb{\rho}_m\right)$, are formulated as~\cite{ITU525}:
\begin{equation}
\label{gain}
\beta_{\mathrm{s}}\left(\mathbf{q}\right)={\beta_0}{\left|\left|\pmb{\rho}_{\mathrm{RIS}}\right|\right|_2^{-2}},~\beta\left(\mathbf{q}, \pmb{\rho}_m\right)={\beta_0}{\left|\left|\pmb{\rho}_{\mathrm{RIS}}-\pmb{\rho}_m\right|\right|_2^{-2}},
\end{equation}
where $\beta_0$ represents the reference path loss at a 1~m distance. 
Therefore, the channel for the source-to-RIS link, $\mathbf{H}\left(\mathbf{q}\right)\in\mathbb{C}^{N\times M}$, and the RIS-to-destination link, $\mathbf{h}^*\left(\mathbf{q},\cdot\right)\in\mathbb{C}^{1\times N}$, 
are expressed as:
\begin{equation}
	\begin{split}
	\label{ULA}
	\begin{cases}
	\mathbf{H}\left(\mathbf{q}\right)\\=\sqrt{\beta_{\mathrm{s}} \left(\mathbf{q}\right)}e^{j\left(\Phi_{\mathbf{H}}-\frac{2\pi\left|\left|\pmb{\rho}_{\mathrm{RIS}}\right|\right|_2}{\lambda}\right)} \mathbf{a}_{\mathrm{RIS}} \left(\phi_{\mathrm{r,RIS}} \left(\mathbf{q}\right)\right) \mathbf{a}_{\mathrm{s}}^* \left(\phi_{\mathrm{t,s}} \left(\mathbf{q}\right)\right)\\
	\mathbf{h}^*\left(\mathbf{q},\pmb{\rho}_m\right)\\=\sqrt{\beta \left(\mathbf{q}, \pmb{\rho}_m\right)}e^{j\left(\Phi_{\mathbf{h}}-\frac{2\pi \left|\left|\pmb{\rho}_{\mathrm{RIS}} - \pmb{\rho}_m\right|\right|_2}{\lambda}\right)} \mathbf{a}_{\mathrm{RIS}}^* \left(\phi_{\mathrm{t,RIS}} \left(\mathbf{q}, \pmb{\rho}_m\right)\right),
\end{cases}
\end{split}
\end{equation}
{where $\Phi_{\mathbf{H}}$ and $\Phi_{\mathbf{h}}$ are independent and uniformly-distributed random phases within $\left[0, 2\pi\right)$.} Therein, the array response $\mathbf{a}_{\mathrm{s}}\left(\cdot\right)$ and $ \mathbf{a}_{\mathrm{RIS}}\left(\cdot\right)\in\mathbb{C}^N$ of the source and aerial-RIS, respectively, are given by:
\begin{equation}
\label{ar}
\begin{split}
\begin{cases}
\mathbf{a}_{\mathrm{s}}\left(\cdot\right)=\left[\{e^{-j2\pi\left(m-1\right)\bar{d}_{\mathrm{s}} \left(\sin\left(\cdot\right)\right)}\}_{m=0}^{M-1}\right]^{\mathrm{T}}\\
\mathbf{a}_{\mathrm{RIS}}\left(\cdot\right)=\left[\{e^{-j2\pi\left(n-1\right)\bar{d} \left(\sin\left(\cdot\right)\right)}\}_{n=0}^{N-1}\right]^{\mathrm{T}},
\end{cases}
\end{split}
\end{equation}
where $\bar{d}_{\mathrm{s}} \triangleq \frac{d_{\mathrm{s}}}{\lambda}$ and $\bar{d}\triangleq\frac{d_{\mathrm{RIS}}}{\lambda}$. Lastly, $\phi_{\mathrm{t,s}}\left(\mathbf{q}\right)$, $\phi_{\mathrm{t,RIS}}\left(\mathbf{q}, \cdot\right)$, and $\phi_{\mathrm{r,RIS}} \left(\mathbf{q}\right)$ denote the {angle-of-departure (AoD)} of the source-RIS and RIS-destination links and the angle-of-arrival (AoA) of the source-RIS link, respectively. For analytical tractability, it is assumed that the ground source has knowledge of channels, 
which can be obtained using methods outlined in~\cite{CE1, CE3, CE2, CE4}.

\subsection{Implementation of Active-RIS}\label{4.ac}
Since conventional RIS consists of passive elements which cannot amplify the signal, the reflected signal has to suffer from ``multiplicative fading'' induced by the multiplication of the path loss of source-to-RIS and RIS-to-UAV-BS links~\cite{aris1}, which is presented by $\left|\left|\pmb{\rho}_{\mathrm{RIS}}^*-\pmb{\rho}_m \right|\right|_2^2 \left|\left|\pmb{\rho}_{\mathrm{RIS}}^*\right|\right|_2^2 $ in~(\ref{RSNR}). Recently, by accompanying the active-RIS~\cite{aris2, aris3, aris4, aris5}, we can solve the aforementioned ``multiplicative fading'' problem by compensating signal power at the active-RIS component.

We assume that, as illustrated in Fig.~\ref{fig_overview}(b), the active amplification circuit with amplification factor $\{\alpha_n\}_{n=1}^N$ and the element-wise hardware power consumption $P_{\mathrm{E}}$ is implemented into the $N$-element aerial-RIS. In other words, after supplying the hardware power consumption for $N$ active elements, we can generate feasible amplification gain by remaining power for a given power budget~\cite{aris1, aris5}. The reflected and amplified signal $\mathbf{y}_{\mathrm{ra}}$ for incident signal $\mathbf{x}_{\mathrm{ra}}$ is modeled by~\cite{aris1}:
\begin{equation} 
\mathbf{y}_{\mathrm{ra}}=\underbrace{\mathbf{A}\pmb\Theta \mathbf{x}_{\mathrm{ra}}}_{\textrm{Desired signal}} +\underbrace{ \mathbf{A}\pmb{\Theta} {\mathbf{n}}_{\mathrm{a}}}_{\textrm{Dynamic noise}} + \underbrace{\mathbf{n}_{\mathrm{s}}}_{\textrm{Static noise}}.
\label{acSM}
\end{equation}
Here, $\mathbf{A}\triangleq\mathrm{diag}\left(\left\{\alpha_n\right\}_{n=1}^N\right)$ is the amplification matrix of the active-RIS, wherein we assume $1<\alpha_n\le\alpha_{\max}~(\forall n)$~\cite{wongjsac}. Moreover, $\pmb{\Theta}\triangleq\mathrm{diag}\left(\left\{e^{j\theta_n}\right\}_{n=1}^N\right)\in\mathbb{C}^{N\times N}$ is a phase shift matrix with phase shift $\theta_n \in \left[0,2\pi\right)$ of the $n$th element, $\mathbf{n}_{\mathrm{a}} \sim \mathcal{CN} \left(\mathbf{0}_N , \sigma_{\mathrm{a}}^2 \mathbf{I}_N \right)$ is the dynamic noise induced by the amplification of the active-RIS elements, and {$\mathbf{n}_{\mathrm{s}}$ is the static additive Gaussian noise uncorrelated to $\mathbf{A}$ and $\mathbf{n}_{\mathrm{a}}$~\cite{aris1, aris2}.} 

By definition, $\mathbf{x}_{\mathrm{ra}}$ is given by $\mathbf{x}_{\mathrm{ra}}=\mathbf{H}\left(\mathbf{q}\right)\mathbf{x}$, where $\mathbf{x}$ is the total transmit signal. We can represent $\mathbf{x}$ as a sum of unit-magnitude precoding vectors $\mathbf{v}_m \in \mathbb{C}^M$, each corresponding to a unit-power signal $s_m$ intended for UAV-BS $m$, with transmit power $P_m$ at the source. Therefore, $\mathbf{x}$ is:
\begin{equation}
\label{tx}
\mathbf{x}=\sum_{m\in \mathbb{M}} \mathbf{v}_m \sqrt{P_m G_{\mathrm{s}} }s_m,
\end{equation}
which implies that the {source transmit} power $P_{\mathrm{tot,s}}$ is:
\begin{equation}
\label{tots}
P_{\mathrm{tot,s}}=\mathcal{E}\left[\left|\left|\pmb{\mathrm{x}}\right|\right|^2\right]=\mathrm{tr}\left(G_{\mathrm{s}} \pmb{\mathrm{PV^*V}}\right)=G_{\mathrm{s}}\sum_{m\in\mathbb{M}} P_m.
\end{equation}
By~(\ref{acSM}), the reflection power of active-RIS is given by~\cite{Doh}
\begin{equation}
\begin{aligned}
\label{refPow}
P_{\mathrm{R}} &=\mathcal{E}\left[ ||\mathbf{A}\pmb\Theta (\mathbf{x}_{\mathrm{ra}}+\mathbf{n}_{\mathrm{a}})||_2^2  \right]
= \mathcal{E}[||\mathbf{A}\mathbf{H}(\mathbf{q})\mathbf{x}||_2^2]+\sigma_{\mathrm{a}}^2 \sum_{n=1}^N \alpha_n^2,
\end{aligned}
\end{equation}
and the power consumed by the active hardware components is expressed as $NP_{\mathrm{E}} \triangleq N(P_{\mathrm{DC}}+P_{\mathrm{SW}})$\cite{aris4}, where $P_{\mathrm{E}}$ denotes the per-element hardware power consumption comprising the control and phase-shift switching power $P_{\mathrm{SW}}$ and the direct current (DC) biasing power $P_{\mathrm{DC}}$ required for the amplifier in each active-RIS element~\cite{dsRIS}. Consequently, the total power consumption of the active-RIS is given by
\begin{equation}
\label{totalACT}
P_{\mathrm{tot,a}} = P_{\mathrm{R}} + NP_{\mathrm{E}}, 
\end{equation}
where we assume that the efficiency of the power amplifier is 1 in both source and aerial-active-RIS. 

Moreover, $\mathbf{y}_{\mathrm{ra}}$ faces the channel $\mathbf{h}^*\left(\mathbf{q},\pmb{\rho}_m\right)$ from RIS-destination. Hence, by concatenating the effects, the signal model from source-to-RIS-to-UAV-BS-$m$ is given by~\cite{Doh}
\begin{equation}
\label{yHxn}
{y_m= \mathbf{h}^*\left(\mathbf{q},\pmb{\rho}_m\right)\mathbf{A} \pmb{\Theta} \mathbf{H}\left(\mathbf{q}\right)\mathbf{x}+\mathbf{h}^*\left(\mathbf{q},\pmb{\rho}_m\right) \mathbf{A}\pmb{\Theta}\mathbf{n_a}+n,}
\end{equation}
where $n\sim \mathcal{CN}(0, \sigma^2 )$ is the noise at the receiver. By~(\ref{yHxn}), the backhaul rate of UAV-BS $m$ is given by~\cite{Doh}
\begin{equation}
	\label{rate}
	R_m=\frac{B_{\mathrm{b}}}{M_0} \log_2 \left(1+\underbrace{\frac{ P_m G_{\mathrm{s}}\left|\mathbf{h}^*\left(\mathbf{q},\pmb{\rho}_m\right) \mathbf{A}\pmb{\Theta} \mathbf{H}\left(\mathbf{q}\right)\mathbf{v}_m\right|^2}{ \sigma_{\mathrm{a}}^2||\mathbf{A}\pmb{\Theta}^*\mathbf{h}\left(\mathbf{q},\pmb{\rho}_m\right)||^2+\sigma^2}}_{\triangleq\gamma_m}\right),
\end{equation}
where $B_{\mathrm{b}}$ denotes the backhaul bandwidth, which is equally partitioned into $M_0$ sub-bands assigned to each UAV-BS, and $\gamma_m$ is the received SNR of UAV-BS $m$. Here, we can notice the definition of $\gamma_m$ that the signal is amplified by $\{\alpha_n\}$ and the denominator is extended with the power of the dynamic noise signal $\mathbf{h}^*\left(\mathbf{q},\pmb{\rho}_m\right) \mathbf{A}\pmb{\Theta}\mathbf{n_a}$. 
Moreover, since $\mathbf{h}^*\left(\mathbf{q},\pmb{\rho}_m\right)$ contains $\beta\left(\mathbf{q}, \pmb{\rho}_m\right)\triangleq\frac{\beta_0}{\left|\left|\pmb{\rho}_{\mathrm{RIS}}-\pmb{\rho}_m\right|\right|_2^2}$, the effect of the dynamic noise gets weaker when the RIS-to-UAV-BS distance gets larger. 

To enhance energy-efficiency by minimizing the source transmit power, i.e., $\sum_{m\in \mathbb{M}} P_m$\cite{Noh, HBRIS}, which will be clarified in~(\ref{etaee}), we first derive the precoding vector $\mathbf{v}_m$ for UAV-BS $m$ that maximizes $\gamma_m$ under a given $P_m$. This approach enables a reduction in transmit power while maintaining the required data rate, {and is formulated by MRT:}
	\begin{equation}
		\label{vm}
		{\mathbf{v}_m=\frac{\mathbf{a}_{\mathrm{s}}\left(\phi_{\mathrm{t,s}} \left(\mathbf{q}\right)\right)}{\sqrt{M}}~\left(\forall m\in\mathbb{M}\right),}
	\end{equation}
{which implies that the optimal transmission strategy $\{\mathbf{v}_m\}_{m\in\mathbb{M}}$ is given to maximize the inner-product with $\mathbf{a}_{\mathrm{s}}^*\left(\phi_{\mathrm{t,s}} \left(\mathbf{q}\right)\right)$~\cite{HBRIS}.} 
Moreover, in~(\ref{rate}) it is clear that
\begin{equation}
\label{Atheta}
||\mathbf{A}\pmb{\Theta}^*\mathbf{h}\left(\mathbf{q},\pmb{\rho}_m\right)||_2^2=||\mathbf{A}\mathbf{h}\left(\mathbf{q},\pmb{\rho}_m\right)||_2^2=\beta\left(\mathbf{q}, \pmb{\rho}_m\right)\sum_{n=1}^N\alpha_n^2,
\end{equation}
{where the first equality comes from the fact that $\pmb{\Theta}\triangleq\mathrm{diag}\left(\left\{e^{j\theta_n}\right\}_{n=1}^N\right)$ is unitary.}
By applying MRT and~(\ref{Atheta}), $\gamma_m$ becomes~\eqref{RSNR}
\begin{figure*}
\begin{equation}
\begin{split}
		\label{RSNR}
	\gamma_m=\bar{\gamma}\left|\mathbf{h}^*\left(\mathbf{q},\pmb{\rho}_m\right) \mathbf{A}\pmb{\Theta} \mathbf{H}\left(\mathbf{q}\right)\frac{\mathbf{a}_{\mathrm{s}}\left(\phi_{\mathrm{t,s}} \left(\mathbf{q}\right)\right)}{\sqrt{M}}\right|^2 =\bar{\gamma} \frac{\biggl|\sum_{n=1}^N \alpha_ne^{j\left(\theta_n +2\pi \left(n-1\right) \bar{d}\left(\sin \left(\phi_{\mathrm{t,RIS}}\left(\mathbf{q}, \pmb{\rho}_m\right)\right)-\sin \left(\phi_{\mathrm{r,RIS}}  \left(\mathbf{q}\right)\right)\right)\right)} \biggr|^2}{\left|\left|\pmb{\rho}_{\mathrm{RIS}}\right|\right|_2^2\left|\left|\pmb{\rho}_{\mathrm{RIS}}-\pmb{\rho}_m\right|\right|_2^2},
\end{split}
\end{equation}
\hrule
\end{figure*}
where $\bar{\gamma} \triangleq\frac{P_m G_{\mathrm{s}}\beta_0^2 M}{ \sigma_{\mathrm{a}}^2  \beta(\mathbf{q}, \pmb{\rho}_m)\sum_{n=1}^N\alpha_n^2+\sigma^2}$~\cite{HBRIS}.
\begin{figure}[t]
	\centering
	\begin{center}
		\includegraphics[width=0.98\columnwidth,keepaspectratio]%
		{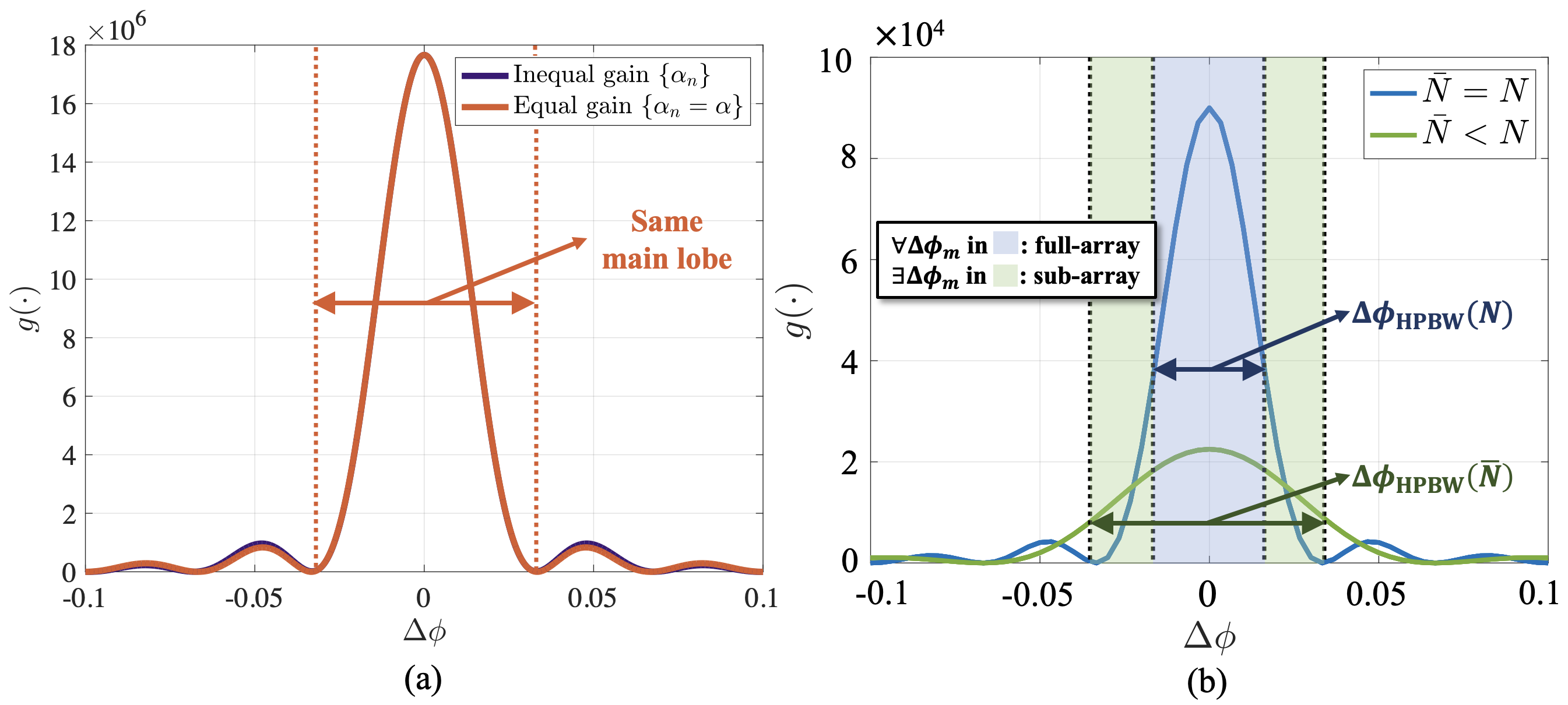}
		\caption{Illustration of passive beamforming gain $g$ and the full/sub-array structure: (a) The main lobe characteristics obtained under a general unequal gain configuration $\{\alpha_n\}_{n=1}^N$ are nearly identical to those under the equal-gain case $\alpha_n=\alpha~(\forall n)$. (b) When the sin-AoD deviation lies beyond the HPBW of the full-array beamforming pattern, a sub-array structure is employed to accommodate the deviated point.}
		\label{figgain}
	\end{center}
\end{figure}

For $\left\{\theta_n^*\right\}_{n=1}^N$ to maximize $\gamma_m$ for a given $P_m$, they must be configured to ensure that the reflected signals are constructively combined at the designated point $\pmb{\rho}_m$:
\begin{equation}
\begin{split}
	\label{phasem}
	&\theta_n^*\left(\mathbf{q},\pmb{\rho}_m\right)\\&=\bar\theta - 2\pi \left(n-1\right) \bar{d} \left(\sin \left(\phi_{\mathrm{t,RIS}}\left(\mathbf{q},\pmb{\rho}_m\right)\right)-\sin\left(\phi_{\mathrm{r,RIS}}  \left(\mathbf{q}\right)\right)\right),
	\end{split}
\end{equation}
where $\bar{\theta}$ is a random phase shift in RIS. However, since $M_0$ UAV-BSs need to be served, the optimal $\left\{\theta_n^*\right\}_{n=1}^N$ vary for each $m\in\mathbb{M}$. Therefore, it is necessary to determine a phase alignment point $\pmb{\bar{\rho}}$ that achieves a Pareto-optimal solution with respect to $\left\{\gamma_m\right\}_{m\in\mathbb{M}}$. That is, for given $\mathbf{q}$ and $\pmb{\bar\rho}$, the phase shifts $\left\{\theta_n^*\right\}_{n=1}^N$ are set as
\begin{equation}
\begin{split}
	\label{phase}
	&\theta_n^*\left(\mathbf{q},\pmb{\bar\rho}\right)\\&=\bar\theta-2\pi \left(n-1\right) \bar{d} \left(\sin \left(\phi_{\mathrm{t,RIS}}\left(\mathbf{q},\pmb{\bar\rho}\right)\right)-\sin \left(\phi_{\mathrm{r,RIS}} \left(\mathbf{q}\right)\right)\right),
\end{split}
\end{equation}
which coherently overlaps the reflected signal to $\pmb{\bar\rho}$. By substituting~(\ref{phase}) into~(\ref{RSNR}), $\gamma_m$ becomes
\begin{equation}
	\label{RSNR2}
	\gamma_m = \bar{\gamma}\frac{\tilde{g}\left(\Delta\phi_m\left(\pmb{\bar\rho}\right)\right)}{ \left|\left|\pmb{\rho}_{\mathrm{RIS}}-\pmb{\rho}_m\right|\right|_2^2 \left|\left|\pmb{\rho}_{\mathrm{RIS}}\right|\right|_2^2},
\end{equation}
where $\tilde{g}\left(\Delta\phi_m \left(\pmb{\bar\rho}\right)\right)$ represents the passive beamforming gain of the aerial-active-RIS towards $\pmb{\rho}_m$, assuming that the phases are aligned with $\pmb{\bar\rho}$. {This gain is derived by evaluating}
\begin{equation}
\begin{aligned}
\label{weightg}
{\tilde{g}\left(\Delta\phi_m \left(\pmb{\bar\rho}\right)\right)=\biggl|\sum_{n=1}^N \alpha_n e^{j\left(2\pi \left(n-1\right) \bar{d}\Delta\phi_m \left(\pmb{\bar\rho}\right)\right)}\biggr|^2}
\end{aligned}
\end{equation}
{in~(\ref{RSNR}), where $\Delta\phi_m \left(\pmb{\bar\rho}\right)$ is the sin-AoD deviation between $\pmb{\bar\rho}$ and $\pmb{\rho}_m$, that is, $\Delta\phi_m \left(\pmb{\bar\rho}\right)\triangleq\sin \left(\phi_{\mathrm{t,RIS}} \left(\mathbf{q}, \pmb{\rho}_m\right)\right)-\sin \left(\phi_{\mathrm{t,RIS}} \left(\mathbf{q}, \pmb{\bar\rho}\right)\right).$ As depicted in Fig.~\ref{figgain}(a), the optimal beamforming gain is achieved within the main lobe of $\tilde{g}$. Moreover, Fig.~\ref{figgain}(a) demonstrates that adopting the equal active-RIS gain scenario yields nearly identical results within the main lobe region. This observation is actually rigorously true, which can be proved by the triangular and Cauchy-Schwartz inequality,
\begin{equation}
\tilde g(\Delta\phi)
\le \left(\sum_{n=1}^N\alpha_n\right)^2
\le N\sum_{n=1}^N\alpha_n^2,
\label{eq:g_upperbound_general}
\end{equation}
respectively, with the first inequality becoming tight when $\Delta\phi$ is sufficiently small (main-lobe region) so that the phasors are nearly aligned. In our deployment, the aerial-active-RIS is placed close to the source and the beam is steered towards $\bar{\boldsymbol{\rho}}$ chosen to represent the Pareto-optimal directions of the UAV-BSs; consequently, $\Delta\phi_m(\bar{\boldsymbol{\rho}})$ stays within the main-lobe vicinity and the variation of the denominator term is minor. In this regime, maximizing $\tilde{g}$ is equivalently achieved by uniform $\alpha_n=\alpha$ by the second inequality, which yields the dominant gain improvement while avoiding an ill-conditioned optimization over element-wise amplitudes that provides marginal additional benefit, as it matches with Fig~\ref{figgain}(a).}

Considering $\alpha_n=\alpha~(\forall n)$ and $\bar{N}\left(\le N\right)$ utilized active-RIS elements, $\tilde{g}$ becomes
\begin{equation}
	\label{beamforming}
	\tilde{g}=\alpha^2 g\left(\Delta\phi_m\left(\pmb{\bar\rho}\right)\right)\triangleq\alpha^2\left|\frac{\sin\left(\pi \bar{N} \bar{d} \Delta \phi_m\left(\pmb{\bar\rho}\right)\right)}{\sin \left(\pi \bar{d} \Delta \phi_m\left(\pmb{\bar\rho}\right)\right)}\right|^2,
\end{equation}
As in Fig.~\ref{figgain}(b), $g$ dissipates to 0 out of its HPBW $\Delta\phi_{\mathrm{HPBW}}$~\cite{HPBW}:
\begin{equation}
\label{hpbwww}
\Delta\phi_{\mathrm{HPBW}}\left(\bar{N}\right)\approx\frac{0.8858}{\bar{N}\bar{d}},~\Delta\phi_{\mathrm{HPBW}} \left(N\right)\triangleq\Delta\phi_{\mathrm{HPBW}},
\end{equation}
and the peak gain of $g$ is ${\bar{N}}^2$. Thus, it is necessary to fine-tune  $\pmb{\bar\rho}$ and determine the maximum value of $\bar{N}$ to locate every UAV-BS within the HPBW, thereby maximizing $\left\{g\left(\Delta\phi_m \left(\pmb{\bar\rho}\right)\right)\right\}_{m\in\mathbb{M}}$.

Moreover, $\bar{\gamma}$ is transformed into
\begin{equation}
\label{gammabareq}
\bar{\gamma} \triangleq\frac{P_m G_{\mathrm{s}}\beta_0^2 M}{ \sigma_{\mathrm{a}}^2  \beta(\mathbf{q}, \pmb{\rho}_m)N\alpha^2 +\sigma^2}.
\end{equation}
and we can also manipulate~(\ref{refPow}) by~(\ref{refPowVV}). 
\begin{figure*}
\begin{equation}
\begin{aligned}
\label{refPowVV}
P_{\mathrm{R}} &=\alpha^2 \left(N\beta_{\mathrm{s}} \mathcal{E}\left[\left|\mathbf{a}_s^*(\phi_{t,s}(\mathbf{q}))\sum_{m\in \mathbb{M}} \frac{\mathbf{a}_s(\phi_{t,s}(\mathbf{q}))}{||\mathbf{a}_s(\phi_{t,s}(\mathbf{q}))||_2} \sqrt{P_m G_{\mathrm{s}} }s_m\right|^2\right]+N\sigma_{\mathrm{a}}^2\right)
=\alpha^2 \left(NM\beta_{\mathrm{s}} G_{\mathrm{s}} \sum_{m\in\mathbb{M}} P_m+N\sigma_{\mathrm{a}}^2\right)\\
\end{aligned}
\end{equation}
\hrule
\end{figure*}
By adding $NP_{\mathrm{E}}$ to~(\ref{refPowVV})~\cite{aris5}, it implies that for given $\alpha^2, G_{\mathrm{s}}, N$, and $P_{\mathrm{E}}$, the maximum power constraint of active-RIS is given by
\begin{equation}
\begin{aligned}
\label{pconstaris}
P_{\mathrm{tot,a}}& = P_{\mathrm{R}} + NP_{\mathrm{E}}\\
& =\alpha^2 \left(NM\beta_{\mathrm{s}} G_{\mathrm{s}} \sum_{m\in\mathbb{M}} P_m+N\sigma_{\mathrm{a}}^2\right)+ NP_{\mathrm{E}}\le P_{\max,\mathrm{a}},
\end{aligned}
\end{equation}
where $P_{\max,\mathrm{a}}$ is the maximum threshold of the active-RIS reflection power. {From~\eqref{pconstaris}, the trade-off between $\alpha^2$ and the consumed power becomes explicit. Increasing $\alpha^2$ directly raises the required $P_{\mathrm{R}}$, which in turn tightens $P_{\max,\mathrm{a}}$. Consequently, the feasible power budget for signal transmission ($\sum_{m\in\mathbb M}P_m$) is reduced, thereby limiting the admissible $\{P_m\}$ and reflecting the inherent trade-off between higher amplification gain and increased aerial-active-RIS power consumption.}


To reliably support the UAV-BSs, the backhaul rate $\{R_m\}_{m\in\mathbb{M}}$ provided by the source should be balanced with the throughput $\{C_m\}_{m\in\mathbb{M}}$ of the fronthaul link: $R_m = C_m~\left(\forall m\in\mathbb{M}\right)$, where the balance between fronthaul and backhaul capacities is critical for optimal network performance with avoidance of bottleneck in B5G/6G wireless network~\cite{38.801, HBRIS}. By $R_m=\frac{B_{\mathrm{b}}}{M_0} \log_2 \left(1+\gamma_m \right)$, $P_m$ must satisfy the following constraint~(\ref{Constraint}).
\begin{figure*}
\begin{equation}
	\label{Constraint}
	\begin{aligned}
		P_m= \alpha^{-2} \left(2^{   \frac{M_0}{B_{\mathrm{b}}} C_m    }-1\right) \frac{  \left(\sigma^2+\alpha^2 \sigma_{\mathrm{a}}^2 N \frac{\beta_0}{\left|\left|\pmb{\rho}_{\mathrm{RIS}}-\pmb{\rho}_m\right|\right|_2^2} \right) \left|\left|\pmb{\rho}_{\mathrm{RIS}}-\pmb{\rho}_m \right|\right|_2^2 \left|\left|\pmb{\rho}_{\mathrm{RIS}}\right|\right|_2^2}{ G_{\mathrm{s}} \beta_0^2 M g\left(\Delta \phi_m\left(\pmb{\bar\rho}\right)\right)}~\left(\forall m\in\mathbb{M}\right)
	\end{aligned} 
\end{equation}
\hrule
\end{figure*}
Moreover, the maximum power budget of the source and aerial-active-RIS is given by~\cite{RISEE, aris5}
\begin{equation}
\label{powerC}
\begin{aligned}
(1):P_{\mathrm{tot,s}}=G_{\mathrm{s}}\sum_{m\in\mathbb{M}} P_m\le P_{\mathrm{max}},~(2):(\ref{pconstaris}),
\end{aligned}
\end{equation}
respectively, wherein $\pmb{\mathrm{V}}\triangleq\left[\pmb{\mathrm{v}}_1 \cdots \pmb{\mathrm{v}}_{M_0}\right]\in\mathbb{C}^{M\times M_0}$ derived by~(\ref{vm}), $\pmb{\mathrm{P}}\triangleq\mathrm{diag}\left(\left\{P_m\right\}_{m\in\mathbb{M}}\right)\in\mathbb{R}^{M_0\times M_0}$ and $P_{\mathrm{max}}$ is the feasible threshold of the source transmit power, respectively.

The definition of energy-efficiency is given by~\cite{RISEE}
\begin{equation}
\label{etaee}
\eta\triangleq\frac{\sum_{m\in\mathbb{M}}C_m}{\underbrace{\sum_{m\in\mathbb{M}}(P_m+P_{\mathrm{UAV-BS},m})+P_{\mathrm{tot,a}}+P_{\mathrm{gBS}}+P_{\mathrm{AP}}}_{\triangleq P_0}},
\end{equation}
where the $\sum_{m\in\mathbb{M}}C_m$ comes from $R_m=C_m~(\forall m\in\mathbb{M})$, and {$P_{\mathrm{AP}}$, $P_{\mathrm{gBS}}$ and $\sum_{m\in\mathbb{M}}P_{\mathrm{UAV-BS},m}$ are the hardware-dissipated power used by the aerial platform that carries active-RIS, ground backhaul source and the UAV-BS $m$, respectively.} Hence, {since the numerator $\sum_{m\in\mathbb M} R_m$ becomes constant since the backhaul rate $\{R_m\}_{m\in\mathbb{M}}$ is balanced with the fronthaul throughput $\{C_m\}_{m\in\mathbb{M}}$}, we can conclude that the energy-efficiency maximization problem is equivalent to the minimization of $P_0$, the total power consumption of the whole system. From the denominator of~(\ref{etaee}), $P_0$ includes $P_{\mathrm{AP}}$, $P_{\mathrm{gBS}}$ and $\sum_{m\in\mathbb{M}}P_{\mathrm{UAV-BS},m}$. {In this paper, we adopt two key assumptions: (i) the transmit amplifiers operate within their linear region, and (ii) the circuit power consumption is independent of the communication rate~\cite{RISEE, actpower1, actpower2}. These assumptions are valid for most practical wireless communication systems~\cite{RISEE, aris4, aris5}, where amplifiers are typically designed to work within the linear portion of their transfer function, and where the hardware power consumption $P_{\mathrm{AP}}$\footnote{{Since we assumed that aerial platform is fixed (hovering), the consumed power, which is in general a function of velocity, becomes constant~\cite{battery}.}}, $P_{\mathrm{BS}}$ and $\{P_{\mathrm{UAV-BS},m}\}_{m\in\mathbb{M}}$ can be treated as constant offsets. Therefore, together with $\sum_{m\in\mathbb M}R_m=\sum_{m\in\mathbb M}C_m$, we exclude them from the energy-efficiency maximization process, and consider}
\begin{equation}
\label{objtot}
{\mathrm{obj}\triangleq \sum_{m\in\mathbb{M}}P_m+P_{\mathrm{tot,a}},}
\end{equation}
{which is total transmit and operation (reflection + hardware components) power of the source and aerial-active-RIS, respectively, as a objective function of the energy-efficiency maximization, which can be formulated by~(\ref{objectACT}).} From now on, we will use term ``total power" as~(\ref{objtot}).
\begin{figure*}
\begin{equation}
		\label{objectACT}
		\begin{aligned}
		& \underset{\mathbf{q}, \pmb{\bar\rho}, \left\{\bar{N}\right\}, \alpha, \left\{P_m\right\}_{m\in\mathbb{M}} }{\texttt{min}}~ \left(1+\alpha^2 NM\beta_{\mathrm{s}} G_{\mathrm{s}}\right)\sum\limits_{m\in\mathbb{M}} P_{m}+\alpha^2 N\sigma_{\mathrm{a}}^2+ NP_{\mathrm{E}}~(\triangleq\mathrm{obj}) \\
		  \text{~~s.t.} ~&P_m= \alpha^{-2} \left(2^{   \frac{M_0}{B_{\mathrm{b}}} C_m    }-1\right) \frac{  \left(\sigma^2+\alpha^2 \sigma_{\mathrm{a}}^2 N \frac{\beta_0}{\left|\left|\pmb{\rho}_{\mathrm{RIS}}-\pmb{\rho}_m\right|\right|_2^2} \right) \left|\left|\pmb{\rho}_{\mathrm{RIS}}-\pmb{\rho}_m \right|\right|_2^2 \left|\left|\pmb{\rho}_{\mathrm{RIS}}\right|\right|_2^2}{ G_{\mathrm{s}} \beta_0^2 M g\left(\Delta \phi_m\left(\pmb{\bar\rho}\right)\right)}(\triangleq\mathrm{RHS_1})~\left(\forall m\in\mathbb{M}\right), \\~~~~~&\sum_{m\in\mathbb{M}} P_m\le \min \left\{ G_{\mathrm{s}}^{-1}P_{\mathrm{max}}, {\alpha^{-2} G_{\mathrm{s}}^{-1}}N^{-1}M^{-1}\beta_{\mathrm{s}}^{-1}\left(P_{\max,\mathrm{a}}-NP_{\mathrm{E}}-\alpha^2  N \sigma_{\mathrm{a}}^2 \right)\right\}(\triangleq\mathrm{RHS_2})
		\end{aligned}	
\end{equation}
\hrule
\end{figure*}
\begin{remark}
		\label{r0}
		Since $C_m>0$, $\left|\left|\pmb{\rho}_{\mathrm{RIS}}\right|\right|_2\ge H>0$ and $\left|\left|\pmb{\rho}_{\mathrm{RIS}} - \pmb{\rho}_m\right|\right|_2 >0$ ($\because \pmb{\rho}_{\mathrm{RIS}}$ is extremely close to the origin, as shown in \emph{Theorem}~\ref{cubiceq} and Fig.~\ref{fig_cubic}). Thereby, the right-hand side (RHS) of~(\ref{Constraint}) is non-zero, thus ensuring non-zero transmit power for every UAV-BS via aerial-active-RIS.
		\end{remark}

\section{Proposed Algorithm}
\subsection{Minimizing the Numerator}
Problem~(\ref{objectACT}) is highly nonlinear and non-convex due to its highly-cluttered $\mathrm{RHS}_1$ and $\mathrm{RHS}_2$ of the constraints. Therefore, we will approach the problem by first assuming that $\alpha$ is given, and minimizing the numerator and maximizing the denominator of $\mathrm{RHS}_1$, respectively, which leads to the minimization of the objective function. After that, we will minimize the total power with respect to $\alpha$. Through numerical simulations, we will show that the second constraint, which represents the upper-bound of the source transmit power, does not impact the feasibility of the problem. Therefore, our approach of focusing primarily on the first constraint is justified for energy-efficiency minimization.

It first leads to minimizing the numerator of $\mathrm{RHS}_1$. For given $\alpha$, if we multiply $1+\alpha^2 NM\beta_{\mathrm{s}} G_{\mathrm{s}}$ to the both sides of the first constraint of~(\ref{objectACT}), which becomes the first term of the objective function in~(\ref{objectACT}) related to $\sum_{m\in\mathbb{M}} P_m$, it becomes~(\ref{modfirst}).
\begin{figure*}
\begin{equation}
\begin{aligned}
\label{modfirst}
&\left(1+\alpha^2 NM\beta_{\mathrm{s}} G_{\mathrm{s}}\right)P_m\\
&= \alpha^{-2} \left(2^{   \frac{M_0}{B_{\mathrm{b}}} C_m    }-1\right) \frac{ \left(1+\alpha^2 NM\frac{\beta_0}{||\pmb{\rho}_{\mathrm{RIS}}||_2^2} G_{\mathrm{s}}\right) \left(\sigma^2+\alpha^2 \sigma_{\mathrm{a}}^2 N \frac{\beta_0}{\left|\left|\pmb{\rho}_{\mathrm{RIS}}-\pmb{\rho}_m\right|\right|_2^2} \right) \left|\left|\pmb{\rho}_{\mathrm{RIS}}-\pmb{\rho}_m \right|\right|_2^2 \left|\left|\pmb{\rho}_{\mathrm{RIS}}\right|\right|_2^2}{ G_{\mathrm{s}} \beta_0^2 M g\left(\Delta \phi_m\left(\pmb{\bar\rho}\right)\right)}~ \left(\forall m\in\mathbb{M}\right)
\end{aligned}
\end{equation}
\hrule
\end{figure*}
Hence, the RHS of~(\ref{modfirst}) is determined by
\begin{equation}
\label{quartmod}
(\left|\left|\pmb{\rho}_{\mathrm{RIS}}\right|\right|_2^2+\tilde{\Omega}_1)(\left|\left|\pmb{\rho}_{\mathrm{RIS}}-\pmb{\rho}_m \right|\right|_2^2 +\tilde{\Omega}_2),
\end{equation}
where
\begin{equation}
\label{omegadef}
\tilde{\Omega}_1 =\alpha^2 NM \beta_0 G_{\mathrm{s}}, \tilde{\Omega}_2 =\alpha^2 \frac{\sigma_{\mathrm{a}}^2 }{\sigma^2}N\beta_0.
\end{equation}
{Thereafter, by letting $\mathbf{q}_m$ the 2D location of the aerial-active-RIS considering only UAV-BS $m~(\pmb{\rho}_{\mathrm{RIS}}=[\mathbf{q}_m^{\mathrm{T}}~H]^{\mathrm{T}})$, the numerator minimization becomes equivalent to:}
		\begin{equation}
			\label{num}
			\begin{aligned}
			&	 \underset{\mathbf{q}_m}{\texttt{min}}~(\left|\left|\pmb{\rho}_{\mathrm{RIS}}\right|\right|_2^2+\tilde{\Omega}_1)(\left|\left|\pmb{\rho}_{\mathrm{RIS}}-\pmb{\rho}_m \right|\right|_2^2 +\tilde{\Omega}_2) \\
			&	=\left(H^2 + \left|\left|\mathbf{q}_m\right|\right|_2^2 +\tilde{\Omega}_1\right)\left(\left(H-h_m\right)^2 + \left|\left|\mathbf{q}_m-\mathbf{w}_m \right|\right|_2^2+\tilde{\Omega}_2\right)\\
			& \text{~~s.t.} ~\left|\left|\mathbf{q}_m\right|\right|_2 \ll \delta \left|\left|\mathbf{w}_m\right|\right|_2,
			\end{aligned}	
		\end{equation}
	where $\delta\ll1$ is constant. The constraint ``$||\mathbf{q}_m||_2 \ll \delta ||\mathbf{w}_m||_2$’’ is imposed to keep $\mathbf{q}_m$ near the source. {Positioning $\mathbf{q}_m$ with fixed $H$} close to the source ensures that the full-array RIS architecture ($\bar{N}=N$) is almost certainly utilized, which maximizes the minimum SNR~\cite{UAVRIS} and consequently reduces the total transmit power $(1+\alpha^2 N M \beta_{\mathrm{s}} G_{\mathrm{s}})\sum_{m\in\mathbb{M}} P_m$. 
	Fortunately, we can find a practical solution for the problem, as stated in \emph{Theorem}~\ref{cubiceq}.
	\begin{theorem}
	\label{cubiceq}
	The solution of Problem~(\ref{num}) is given by
	\begin{equation}
	\label{Ps}
	\mathbf{q}_m^* = \kappa_m \mathbf{w}_m ,
	\end{equation}
	where
	\begin{equation}
	\begin{split}
	\label{xixixixi}
	&\kappa_m=\frac{1}{2}+2\sqrt{-\frac{a}{3}}\cos \left( \frac{1}{3} \cos^{-1} \left( \frac{3b}{2a} \sqrt {-\frac{3}{a}} \right)-\frac{4}{3}\pi  \right).
	\end{split}
	\end{equation}
	Here, $a$ and $b$ are given by
	\begin{equation}
	\label{xixixixixixi}
	a\triangleq\frac{1}{2} \left(\zeta_1^2 + \zeta_2^2 +\bar{\Omega}_1+\bar{\Omega}_2 \right) - \frac{1}{4}, b\triangleq\frac{1}{4} \left(\zeta_2^2 - \zeta_1^2 +\bar{\Omega}_2 - \bar{\Omega}_1\right),
		\end{equation}
		where
		\begin{equation}
		\label{ab}
		\zeta_1\triangleq\frac{H}{\left|\left|\mathbf{w}_m\right|\right|_2}, \zeta_2 \triangleq \frac{\left|H-h_m\right|}{\left|\left|\mathbf{w}_m\right|\right|_2}, \bar{\Omega}_i\triangleq ||\mathbf{w}_m||_2^{-2}\tilde{\Omega}_i~(i=1,2).
	\end{equation}
		\end{theorem}
		\begin{proof}
See Appendix A.
	\end{proof}
We numerically confirm in Section IV that $\kappa_m$ remains positive yet very close to zero under the given assumptions, thereby validating that $\mathbf{q}_m^* = \kappa_m \mathbf{w}_m$ serves as an appropriate solution for the proposed power-minimization procedure.

 After obtaining $\left\{\mathbf{q}_m^*\right\}_{m\in\mathbb{M}}$, we need to determine a single $\mathbf{q}^*$ that achieves a Pareto-optimal for all $m\in\mathbb{M}$ in relation to~(\ref{num}). We adopt a global criterion approach that minimizes the total $\ell_2$-distances, thereby guiding the solution toward the Pareto front~\cite{MVO}.
\begin{equation}
	\label{devmin}
	\begin{aligned}
		&	 \underset{\mathbf{q}}{\texttt{min}}~\sum_{m\in \mathbb{M}} \left|\left|\mathbf{q}_m^*-\mathbf{q} \right|\right|_2 .
	\end{aligned}	
\end{equation}
Problem~(\ref{devmin}) is known as the Fermat-Torricelli problem, which is convex and can therefore be efficiently solved using Weiszfeld's algorithm~\cite{FT} with guaranteeing the optimal solution~\cite{WF2}. Consequently, by solving~(\ref{devmin}), we obtain $\mathbf{q}^*$, which serves as a suboptimal solution for the numerator minimization. {The obtained $\mathbf{q}^*$ provides a clear deployment insight: the aerial-active-RIS should be placed sufficiently closer to the source in terms of 2D distance to enhance the incident signal power, while maintaining a moderate distance from the UAV-BSs to balance the overall power consumption.}
\begin{remark}
\label{r2}
To keep $||\mathbf{q}^*||_2$ small even in the presence of an outlier within $\{\mathbf{q}_m^*\}_{m\in\mathbb{M}}$, we minimize the sum of norms rather than the squared norms in~(\ref{devmin}), which further enhances the robustness of the solution~\cite{MVO}.
\end{remark}
\subsection{Maximizing the Denominator}
We adopt the same methodology as described in Section III.B of~\cite{UAVRIS} to solve the denominator~(\ref{objectACT}) since the denominator is exactly same with the passive-RIS-implemented scenario in~\cite{UAVRIS}. This allows us to determine the sub-optimal values for $\{\bar{N}\}$, $\pmb{\bar\rho}~(\mathrm{or}~\{\pmb{\bar\rho}_i^*\}_{i=1}^L)$, and $\pmb{\Theta}$ for utilizing both full RIS array or $L$-times partitioned sub-array scenario.
 \subsection{Determining $\alpha$ and $\{P_m\}_{m\in\mathbb{M}}$}
 Since we have determined $\mathbf{q}^*, \{\bar{N}\}, \pmb{\bar\rho}~(\mathrm{or}~\{\pmb{\bar\rho}_i^*\}_{i=1}^L), \pmb\Theta$, the optimal $\alpha^*$ is determined by following \emph{Theorem}~\ref{thmalpha}.
 
 	\begin{theorem}
	\label{thmalpha}
	For given $\mathbf{q}^*, \{\bar{N}\}, \pmb{\bar\rho}~(\mathrm{or}~\{\pmb{\bar\rho}_i^*\}_{i=1}^L), \pmb\Theta$, the optimum $\alpha$ that minimizes $\mathrm{obj}$ is given by
\begin{equation}
\label{optalphathm}
\alpha^*=\min\left\{\sqrt[4]{\frac{\sum_{m\in\mathbb{M}}\Omega_0\left|\left|\pmb{\rho}_{\mathrm{RIS}}\right|\right|_2^2 \left|\left|\pmb{\rho}_{\mathrm{RIS}}-\pmb{\rho}_m \right|\right|_2^2 }{N\sigma_{\mathrm{a}}^2+\sum_{m\in\mathbb{M}}\Omega_0{\Omega_1}{\Omega_2}}}, \alpha_{\max}\right\},
\end{equation}
 where
 \begin{equation}
 \begin{aligned}
 \label{omegaexplain}
 &\Omega_0=\sigma^{2}\frac{2^{   \frac{M_0}{B_{\mathrm{b}}} C_m    }-1}{G_{\mathrm{s}} \beta_0^2 M g\left(\Delta \phi_m^*\right)}\\
 &{\Omega_1}=\alpha^{-2}\tilde{\Omega}_1= NM \beta_0 G_{\mathrm{s}},~{\Omega_2}=\alpha^{-2}\tilde{\Omega}_2=\frac{\sigma_{\mathrm{a}}^2 }{\sigma^2}N\beta_0.
 \end{aligned}
 \end{equation}
		\end{theorem}
		\begin{proof}
See Appendix B.
	\end{proof}

Since we determine all the variables $\mathbf{q}, \pmb{\bar\rho}, \left\{\bar{N}\right\}, \alpha^*$, the power $\{P_m\}_{m\in\mathbb{M}}$ is given by $\mathrm{RHS}_1$, which can be expressed by~(\ref{optP}), where $\pmb{\rho}_{\mathrm{RIS}}^*=\left[\mathbf{q}^{* \mathrm{T}}~H\right]^{\mathrm{T}}$ and $\Delta\phi_m^*$ is defined by
\begin{equation}
\label{finalphi}
\Delta\phi_m^* \triangleq
\begin{cases}
\Delta\phi_m\left(\pmb{\bar\rho}^*\right)~\left(\bar{N}=N~(\text{full-array})\right) \\
\Delta\phi_m \left(\pmb{\bar\rho}_i^*\right)~\left(\bar{N}<N~(\text{sub-array}, m\in\mathbb{M}_i)\right),
\end{cases}
\end{equation}
where the full- and $L$-times-partitioned sub-array scenario with phase-align points $\{\pmb{\bar\rho}_i^*\}_{i=1}^L$ and partitions of UAV-BS $\{\mathbb{M}_i\}_{i=1}^L (\cup_i\mathbb{M}_i=\mathbb{M})$ are defined in Section III.B in~\cite{HBRIS}. Thus, by substituting~(\ref{optP}) and $\alpha^*$ to $\mathrm{obj}$, we can achieve the optimal total power by utilizing aerial-active-RIS.
\begin{figure*}
\begin{equation}
	\label{optP}
	P_m ^*= \alpha^{*-2}\left(2^{   \frac{M_0}{B_{\mathrm{b}}} C_m    }-1\right) \frac{  \left(\sigma^2+\alpha^{*2} \sigma_{\mathrm{a}}^2 N \frac{\beta_0}{\left|\left|\pmb{\rho}_{\mathrm{RIS}}^*-\pmb{\rho}_m\right|\right|_2^2} \right) \left|\left|\pmb{\rho}_{\mathrm{RIS}}^*-\pmb{\rho}_m \right|\right|_2^2 \left|\left|\pmb{\rho}_{\mathrm{RIS}}^*\right|\right|_2^2}{ G_{\mathrm{s}} \beta_0^2 M g\left(\Delta \phi_m^*\right) }~\left(\forall m\in\mathbb{M}\right)
\end{equation}
\hrule
\end{figure*}

\subsection{Conditions for Feasibility}
Although the power is to be positive in \textit{Remark}~\ref{r0}, we cannot claim the strict feasibility since we cannot always guarantee the upper-bound of $\sum_{m\in\mathbb{M}} P_m$ by~(\ref{optP}). Specifically, the backhaul-rate-ensuring constraint leads to:
\begin{equation}
\label{rec}
R_m= C_m \rightarrow P_m = [\mathrm{RHS}_1]_m~(\forall m\in\mathbb{M}),
\end{equation}
where $[\mathrm{RHS}_1]_m$ is the $\mathrm{RHS}_1$ corresponding to $P_m$ in the first constraint of~(\ref{objectACT}), and the source/active-RIS power constraint is trivially given by
\begin{equation}
\label{pc}
\sum_{m\in\mathbb{M}}P_m\le \mathrm{RHS}_2.
\end{equation}
From~(\ref{rec}) and~(\ref{pc}), it is clear that although the individual transmit power meets the constraint in~(\ref{rec}) by choosing the lower-bound itself as a transmit power value, we cannot guarantee $\sum_{m\in\mathbb{M}}[{\mathrm{RHS}_1}]_m\le\mathrm{RHS_2}$~\cite{boyd, RISEE, Din}. {We, however, offer some additional clarification on whether the feasibility holds or not. Using the closed-form rate-matching transmit power in~\eqref{optP}, the total source transmit power can be expressed as
\begin{equation}
\label{eq:rb_sumPm_split}
\sum_{m\in\mathbb M} P_m^*=\frac{d_s^2}{G_s\beta_0^2 M}
\Bigg(
\alpha^{-2}\sigma^2 \sum_{m\in\mathcal M}\frac{\Gamma_m d_m^2}{g_m}
+\sigma_a^2N\beta_0 \sum_{m\in\mathcal M}\frac{\Gamma_m}{g_m}
\Bigg),
\end{equation}
where $\Gamma_m\triangleq 2^{\frac{M_0}{B_b}C_m}-1$, $d_s^2=\|\boldsymbol\rho_{\mathrm{RIS}}^*\|_2^2$, $d_m^2=\|\boldsymbol\rho_{\mathrm{RIS}}^*-\boldsymbol\rho_m\|_2^2$, and $g_m=g(\Delta\phi_m^*)$. Substituting~\eqref{eq:rb_sumPm_split} into the two terms in the left-hand side in~\eqref{objectACT} yields two sufficient feasibility conditions with respect to $\mathrm{RHS}_2$.}

{\textbf{(1. Source power budget $P_{\max}$)}
\begin{equation}
\label{eq:rb_box1_sub}
P_{\max}
\ge
\frac{d_s^2}{\beta_0^2 M}
\Bigg(
\alpha^{-2}\sigma^2 \sum_{m\in\mathcal M}\frac{\Gamma_m d_m^2}{g_m}
+\sigma_a^2N\beta_0 \sum_{m\in\mathcal M}\frac{\Gamma_m}{g_m}
\Bigg).
\end{equation}}

{\textbf{(2. Aerial-Active-RIS power budget $P_{\max,a}$)}
\begin{equation}
\label{eq:rb_box2_sub}
\begin{aligned}
P_{\max,a}\ge&
NP_E+\alpha^2N\sigma_a^2\\
&+N\beta_s\frac{d_s^2}{\beta_0^2}\Bigg(\sigma^2 \sum_{m\in\mathcal M}\frac{\Gamma_m d_m^2}{g_m}
+\alpha^2\sigma_a^2N\beta_0 \sum_{m\in\mathcal M}\frac{\Gamma_m}{g_m}
\Bigg).
\end{aligned}
\end{equation}
Equations~\eqref{eq:rb_box1_sub} and~\eqref{eq:rb_box2_sub} provide an explicit post-solution feasibility check: 
\begin{itemize}
\item Once the optimal variables are obtained, $\sum_m P_m^*$ is evaluated via~\eqref{eq:rb_sumPm_split}
\item The constraint $\sum_m P_m^*\le\mathrm{RHS}_2$ is verified by checking~\eqref{eq:rb_box1_sub} and~\eqref{eq:rb_box2_sub}.
\end{itemize}
Equation~\eqref{eq:rb_sumPm_split} reveals that $\sum_m P_m^*$ consists of two components:
an $\alpha^{-2}$-decreasing term associated with thermal noise and an $\alpha$-independent floor
term induced by the amplified dynamic noise. Consequently, increasing $\alpha$ does not cause
$\sum_m P_m^*$ to diverge. Furthermore, when the phase-alignment/partitioning design keeps all UAV-BSs within the RIS main lobe, the passive beamforming gain satisfies $g_m\simeq N^2$~\cite{HBRIS}. In this regime,
\begin{equation}
\sum_{m}\frac{\Gamma_m}{g_m}\approx \frac{1}{N^2}\sum_m\Gamma_m,~\sum_{m}\frac{\Gamma_m d_m^2}{g_m}\approx \frac{1}{N^2}\sum_m\Gamma_m d_m^2,
\label{smgm}
\end{equation}
so that the last term of the right-hand side in both~\eqref{eq:rb_box1_sub} and~\eqref{eq:rb_box2_sub} scale down with $N$. In Section IV (Fig.~\ref{fig_feas}), we verify that the practical transmit power and total power obtained by the proposed algorithm remain well below the feasible threshold, achieving a feasibility rate of 100\% even with conservative power budget. This confirms the almost-sure feasibility of our approach, thereby rendering the consideration of infeasible scenarios negligible.}
	 \begin{center}
	\begin{table}[t] 
	\centering
		\caption{Simulation Parameters}
		\begin{tabular}{|>{\centering } m{1.8cm} |>{\centering} m{3.9cm} |>{\centering} m{1.8cm} | }
			\hline
			\textbf{Paramet{\tiny }er} & \textbf{Description} & \textbf{Value}
			\tabularnewline
			\hline
			\centering			$B_{\mathrm{b}}$  & Bandwidth of the backhaul link\\(unless referred) & 50 (MHz)  \tabularnewline \hline
			\centering			$M_0$  & Number of UAV-BSs & $\gtrapprox4$~\cite{Noh}  \tabularnewline \hline
			\centering			$\mathcal{G}$  & Targeted urban region& $500\times 500$ (m)  \tabularnewline \hline
			\centering			$\pmb{\rho}_{\mathcal{G}}$  & Center of $\mathcal{G}$ (unless referred) &$[1000 ~0]^{\mathrm{T}}$ (m)  \tabularnewline \hline
			\centering			$(P_{\max}, P_{\max,\mathrm{a}})$  & Feasible threshold of\\source transmit power and active-RIS power consumption& 20 (dBm)~\cite{aris5} \tabularnewline \hline			
			\centering			$G_{\max}$  & Maximum directional gain & 8 (dB) \tabularnewline \hline
			\centering		$H$ & Height of aerial-RIS\\(unless referred) & 180 (m) \tabularnewline \hline
			\centering		$N$ & Number of RIS elements\\(unless referred) & 300 \tabularnewline \hline
			\centering		$M$ & Number of source antennas & 16 \tabularnewline \hline
			\centering      $\left(d_{\mathrm{s}}, d_{\mathrm{RIS}}\right)$ & Source antenna and\\RIS element separations & $\left(\frac{\lambda}{2}, \frac{\lambda}{10}\right)$~\cite{meta} \tabularnewline \hline
			\centering		$\alpha_{\max}^2$ & Maximum amplification gain of the active-RIS & 40~(dB)~\cite{aris5} \tabularnewline \hline
		\centering		$P_{\mathrm{E}}$ & Power consumed by\\single active hardware component & -3.8~(dBm)\\~\cite{aris4, aris5} \tabularnewline \hline
			\centering		$\Delta_0$ & Distance from $\alpha^{*2}$ in dB for non-optimality comparison & -5~(dB) \tabularnewline \hline	
			\centering		$\sigma_{\mathrm{a}}^2$ & Dynamic noise of the\\active-RIS elements & -80~(dBm)\\~\cite{aris4, aris5} \tabularnewline \hline
			\end{tabular}
		\label{SimPar}
	\end{table}
\end{center}

		\subsection{Analysis on Computational Complexity}
The computational complexity of the proposed algorithm is divided into three main stages. In Section III.A, we begin by selecting $\left\{\kappa_m\right\}_{m\in\mathbb{M}}$ according to~(\ref{xixixixi}), which involves a complexity of $\mathcal{O}\left(M_0\right)$. Subsequently, we solve for $\mathbf{q}^*$ using~(\ref{devmin}), which has an upper-bounded complexity of $\mathcal{O}\left(I_{\mathbb{M}} M_0\right)$, where $I_{\left(\cdot\right)}$ represents the number of iterations required by Weiszfeld's algorithm for the given set~\cite{HBRIS, FT, WF2}. In Section III.B, the complexity is identical to that presented in Section III.B of~\cite{HBRIS}, and is given by $\mathcal{O}((L_{\max} + I_L + M_0) M_0 + L)$. Here, $L_{\max}$ denotes the maximum candidate for $L$, identified through a one-dimensional search, and $I_{L} \triangleq \max_{i\in\{1,\cdots,L\}} I_{\mathbb{M}_i}$. The RIS phase alignment step, performed via~(\ref{phase}), incurs a complexity of $\mathcal{O}(N)$. Lastly, in Section~III.C, we optimize $\alpha^*$ and $\left\{P_m^*\right\}_{m\in\mathbb{M}}$ through~(\ref{optalphathm}) and~(\ref{optP}), each requiring $\mathcal{O}\left(M_0\right)$. Thus, the overall computational complexity is bounded as
\begin{equation}
\label{FC}
\begin{split}
&\mathcal{O}\left(\left(I_{\mathbb{M}} + L_{\max} + I_{L} + M_0\right) M_0 + L + N\right) \\
&\approx \mathcal{O}\left(\left(I_{\mathbb{M}} + I_{L} + M_0\right) M_0 + N\right)~\left(\because I_{\left(\cdot\right)} > L_{\max} \ge L\right),
\end{split}
\end{equation}
which is within quadratic order. Therefore, we can conclude that the proposed algorithm is both energy-efficient and computationally efficient.

	\begin{figure}[t]
	\begin{center}
		\includegraphics[width=0.9\columnwidth,keepaspectratio]%
		{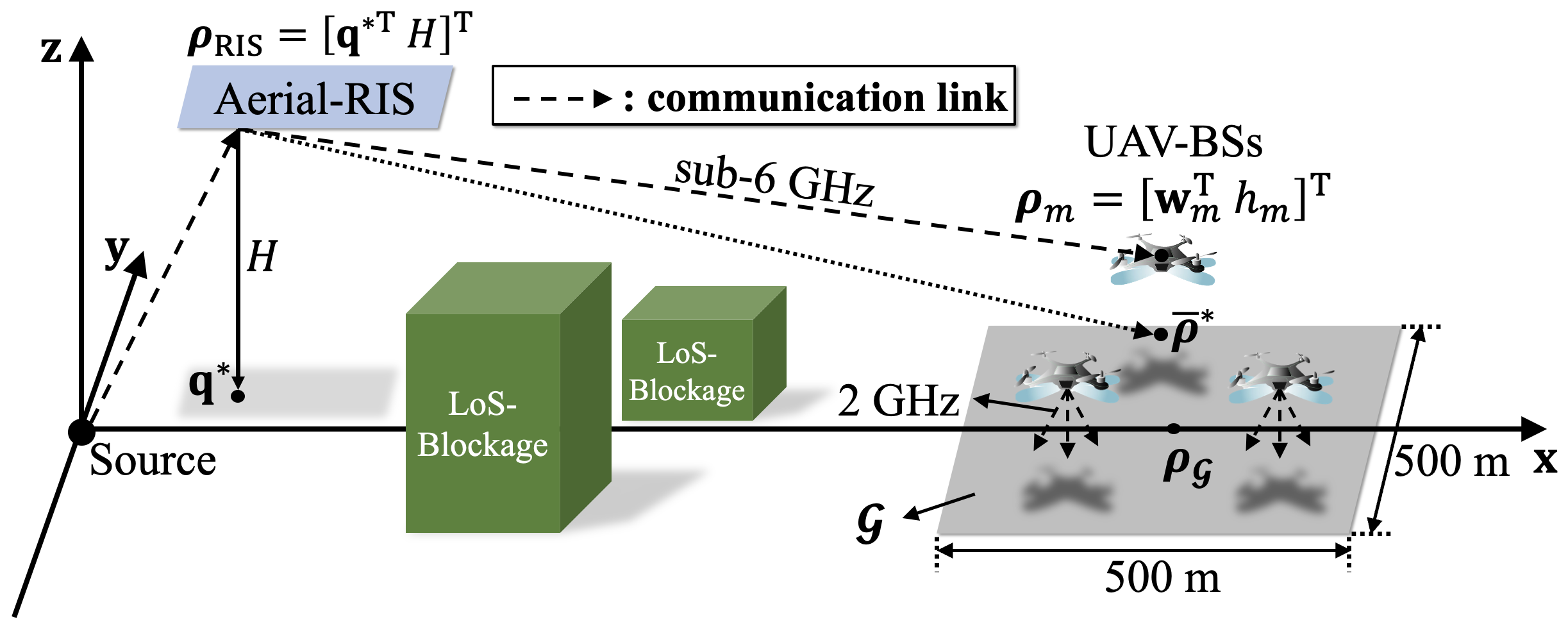}
		\caption{Simulated aerial-active-RIS configuration with an $N$-element active-RIS and $M_0$ UAV-BSs.}
		\label{fig_sim}
	\end{center}
\end{figure}
	\begin{figure}[t]
	\begin{center}
		\includegraphics[width=0.7\columnwidth,keepaspectratio]%
		{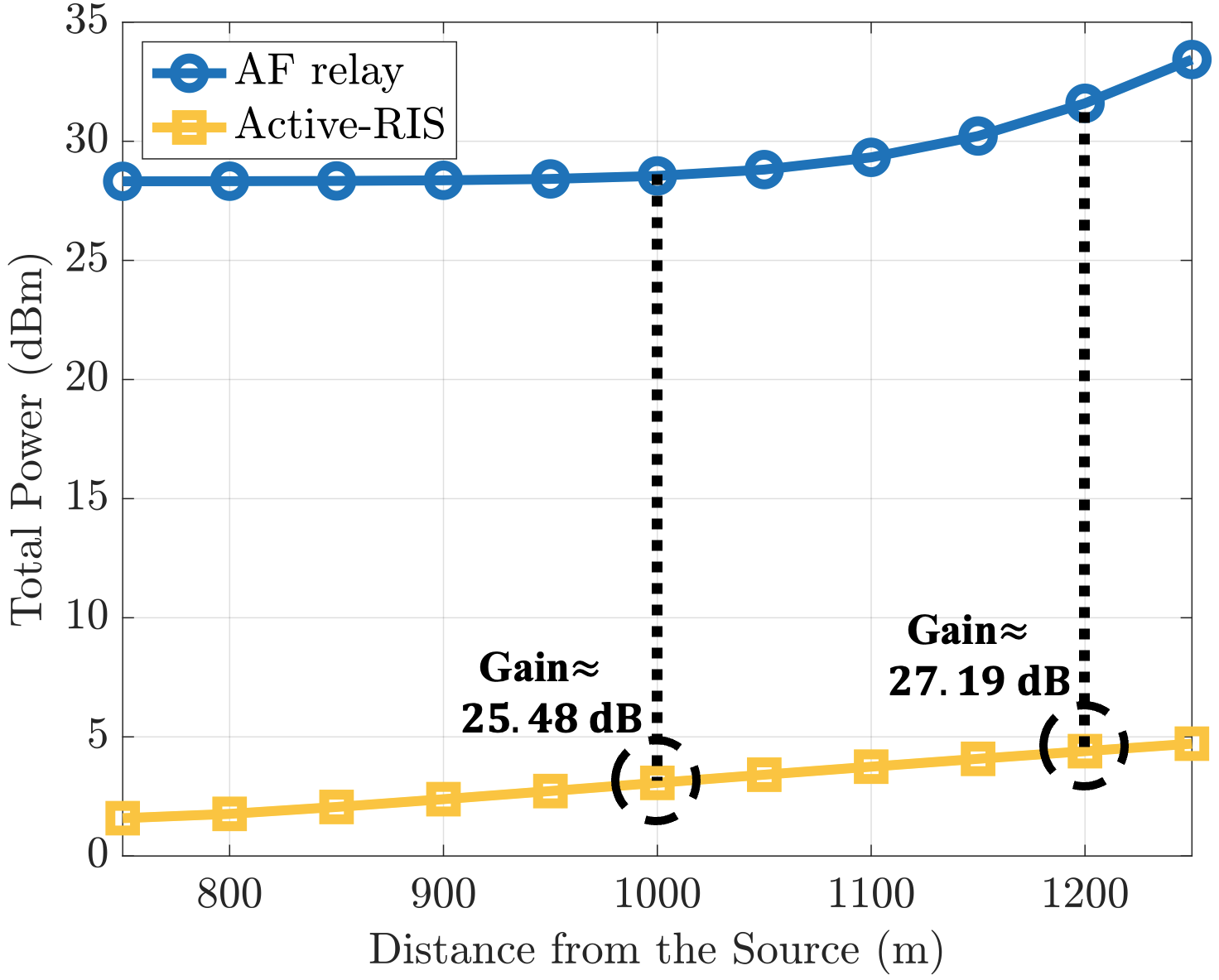}
		\caption{Simulated aerial-active-RIS configuration compared to $M$-array aerial-AF-relay.}
		\label{fig_af}
	\end{center}
\end{figure}
\begin{figure}[t]
	\begin{center}
		\includegraphics[width=0.75\columnwidth,keepaspectratio]%
		{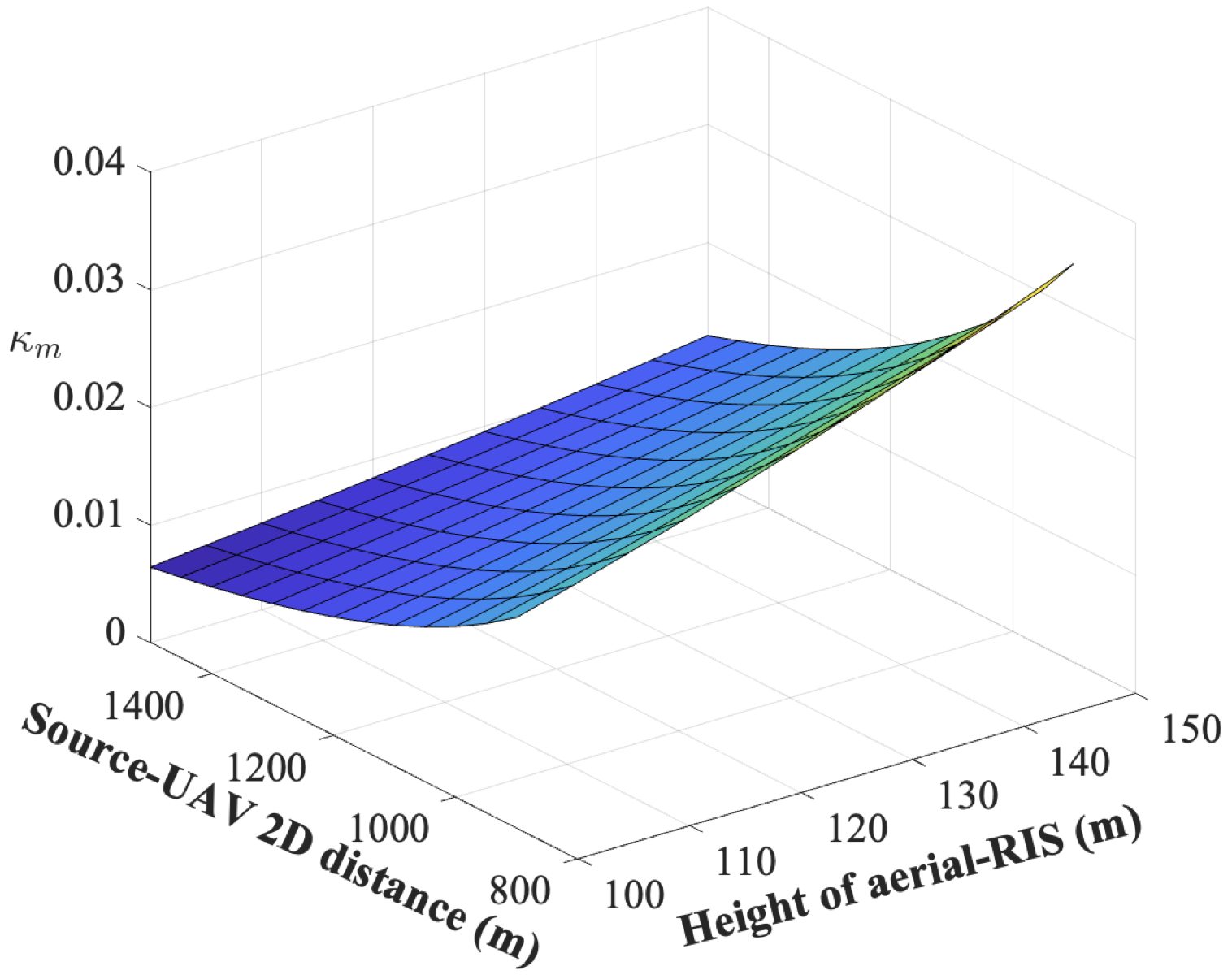}
		\caption{Variation of $\kappa_m$ as a function of the aerial-active-RIS altitude and the 2D source-UAV distance with $h_m = 45~\textrm{m}$.}
		\label{fig_cubic}
	\end{center}
\end{figure}
\section{Numerical Results}\label{4.4.4}
\begin{figure}[t]
	\begin{center}
		\includegraphics[width=0.98\columnwidth,keepaspectratio]%
		{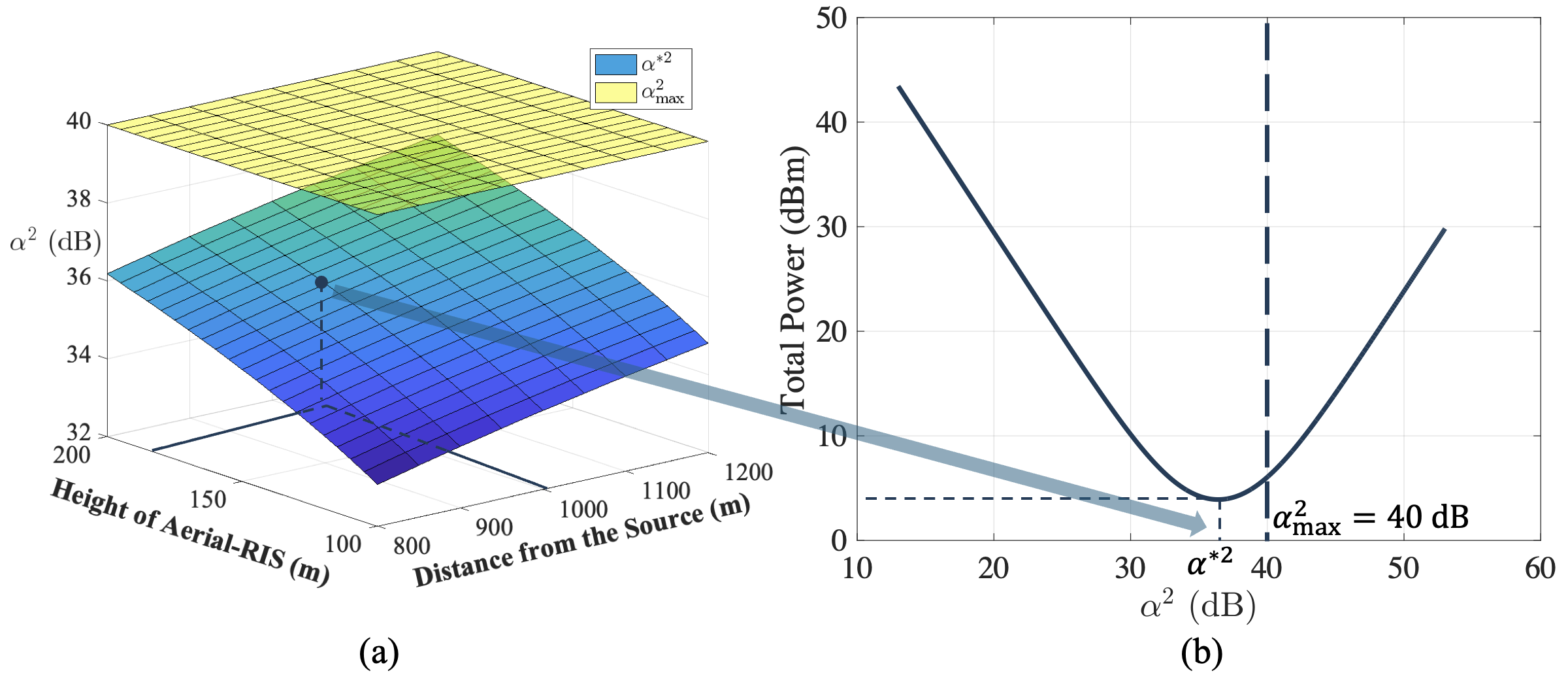}
		\caption{Behavior of (a) $\alpha^{*2}$ with $\alpha_{\max}^2=40$~dB with respect to the height of aerial-RIS and the source-UAV 2D distance and (b) optimal total power with respect to $\alpha^2$ with minimum at $\alpha^{*2}$.}
		\label{fig_alpha}
	\end{center}
\end{figure}
\subsection{Simulation Setup}
We considered $10^3$ independent realizations of randomly distributed users and their associated UAV-BSs~\cite{Noh}. The fronthaul and backhaul links were assumed to operate over 2~GHz and sub-6~GHz frequency bands, respectively~\cite{NR}. Furthermore, the directional antennas at source were assumed to follow the radiation pattern described in~\cite{NR}. Accordingly, the directional antenna gain $G_{\mathrm{s}}(\theta,\phi)$ with maximum directional gain $G_{\max}$ can be expressed as	 
\begin{equation}
\label{ap}
G_{\mathrm{s}}\left(\theta, \phi\right)=G_{\max} -\min\left(A_{\mathrm{v}} \left(\theta\right)+A_{\mathrm{h}} \left(\phi\right) , A_{\max}\right).
\end{equation}
where the vertical and horizontal attenuations $A_{\mathrm{v}}$ and $A_{\mathrm{h}}$, respectively, are given by
\begin{equation}
\begin{split}
\label{antenna}
\begin{cases}
A_{\mathrm{v}} \left(\theta\right) = \min \left(12 \left( \frac{\theta-90^{\circ}}{\theta_{\mathrm{H}}} \right)^2 ,~\mathrm{SLA}_{\mathrm{v}} \right)\\
A_{\mathrm{h}} \left(\phi\right) = \min \left(12 \left( \frac{\phi}{\phi_{\mathrm{H}}} \right)^2 ,~A_{\max} \right),
\end{cases}
\end{split}
\end{equation}
where $\theta \in [0^{\circ}, 180^{\circ}]$ and $\phi \in [-180^{\circ}, 180^{\circ})$ denote the vertical and horizontal angles, $\theta_{\mathrm{H}}$ and $\phi_{\mathrm{H}}$ represent the HPBWs in the vertical and horizontal domains, respectively, and $\mathrm{SLA}_{\mathrm{v}}$ and $A_{\max}$ denote the vertical side-lobe and maximum attenuations, respectively. 

Under this setup, we conducted a numerical performance comparison between the proposed aerial-active-RIS scheme and the passive-RIS-based algorithm in~\cite{HBRIS}, which aims to minimize the transmit power at the source.
Under same achievable rate $\{C_m\}_{m\in\mathbb{M}}$, the energy-efficiency $\eta_{\mathrm{p}}$ for passive-RIS is given by~\cite{RISEE, HBRIS}
\begin{equation}
\label{etaP}
\eta_{\mathrm{p}}=\frac{\sum_{m\in\mathbb{M}}C_m}{\sum_{m\in\mathbb{M}}(P_{m,\mathrm{p}}+P_{\mathrm{UAV-BS},m})+P_{\mathrm{gBS}}+P_{\mathrm{AP}}},
\end{equation}
where $P_{m,\mathrm{p}}$ is the source transmit power corresponds to UAV-BS $m$ with aerial-passive-RIS, which also makes the backhaul rate $\{C_m\}_{m\in\mathbb{M}}$ for each UAV-BS $m$~\cite{HBRIS}. Thus, it is reasonable to compare $\sum_{m\in\mathbb{M}} P_{m,\mathrm{p}}$ (in~\cite{HBRIS}) with $\mathrm{obj}$ for aerial-passive/active-RIS scenarios, respectively, which allows for a fair comparison of energy-efficiency minimization between the two cases. The simulation environment based on the parameters is illustrated in Fig.~\ref{fig_sim}, and the detailed parameters are given in Table~\ref{SimPar}.
\begin{figure*}[t]
	\begin{center}
		\includegraphics[width=1.98\columnwidth,keepaspectratio]%
		{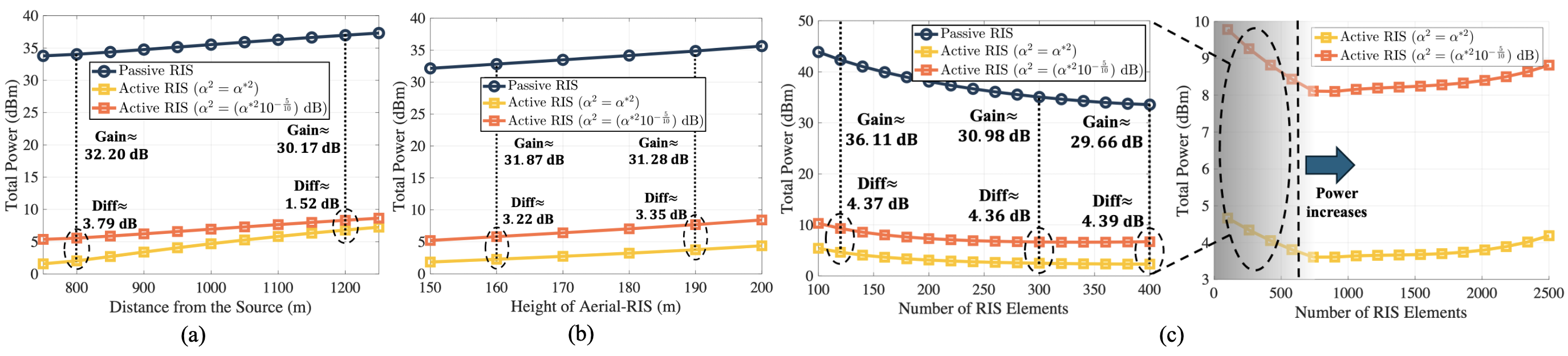}
		\caption{Total power with respect to (a) the distance of $\pmb{\rho}_{\mathcal{G}}$ from the source (b) the height of aerial-RIS (c) the number of RIS elements under the implementation of proposed aerial-active-RIS with benchmarks.}
		\label{actCenter}
	\end{center}
\end{figure*}
\subsection{Active-RIS vs. AF Relay: Performance Comparison}
It is reasonable to explore the scenario where the aerial-active-RIS is not deployed. In such a case, an AF relay on the aerial platform serves as a natural alternative, as it similarly receives the incoming signal and forwards an amplified one to the destination. However, a closer comparison reveals that the proposed aerial-active-RIS architecture can offer greater advantages over the aerial-AF-relay. {Specifically, the total AF-relay power consumption is expressed as~\cite{cute}
\begin{equation}
\label{tafp}
\begin{aligned}
P_{\mathrm{tot,AF}}=& \sum_{m\in\mathbb M} (P_{m,\mathrm{AF}}+P_{\mathrm{UAV-BS}-m})+ P_{\mathrm{R,AF}}\\
&+ P_{\mathrm{gBS}}+ P_{\mathrm{circ,AF}}+P_{\mathrm{AP}},
\end{aligned}
\end{equation}
where $\sum_{m\in\mathbb M}P_{m,\mathrm{AF}}$ is the source transmit power needed to support the target rate, $P_{\mathrm{R,AF}}$ is the relay transmit power (amplification power) required under the AF protocol, $P_{\mathrm{circ,AF}}$ is hardware power cost of the $N$-element AF relay~\cite{cute}:
\begin{equation}
\label{afpower}
P_{\mathrm{circ,AF}}=N(P_{\mathrm{DAC}} + P_{\mathrm{mix}} + P_{\mathrm{filt}}) + P_{\mathrm{syn}},
\end{equation}
where $P_{\mathrm{DAC}}, P_{\mathrm{mix}}, P_{\mathrm{filt}}$, and $P_{\mathrm{syn}}$ represent the power consumed by the digital-to-analog (DAC) converter, mixer, filter, and frequency synthesizer, respectively. Herein, same with the aerial-active-RIS, $P_{\mathrm{AP}}$, $P_{\mathrm{BS}}$ and $\{P_{\mathrm{UAV-BS},m}\}_{m\in\mathbb{M}}$ are treated as constant offsets. } According to practical values reported in~\cite{cute}, the combined term $P_{\mathrm{DAC}} + P_{\mathrm{mix}} + P_{\mathrm{filt}}$ reaches approximately 18.2~dBm, which is over 20~dB higher than the corresponding $P_{\mathrm{E}}$ used in Table~\ref{SimPar}. Furthermore, the additional contribution from $P_{\mathrm{syn}}$ further exacerbates the total power requirement. Although the aerial-AF-relay operates in full-duplex mode, it incurs additional hardware complexity to suppress self-interference~\cite{aris1, RISRe}. As a result, the aerial-AF-relay suffers from significantly higher power demands, particularly on the relay side, making it substantially less energy-efficient than the proposed aerial-active-RIS architecture. 

{For numerical comparison, Fig.~\ref{fig_af} plots the resulting total power versus the distance of $\pmb{\rho}_{\mathcal{G}}$ from the source $\left(d_{\mathcal{G}}\right)$ for aerial-active-RIS vs. $N=M$-element aerial-AF-relay~\cite{af1, af2, af3}\footnote{{Deploying an AF-relay with an antenna array whose size scales with $N$ would imply hundreds of antenna elements, leading to prohibitive hardware cost and power consumption~\cite{cute}. Consequently, the majority of the literature considers AF-relays whose antenna dimensions follow the scale of the TRx~\cite{af1, af2, af3}.}} with optimal configuration employed by numerical exhaustive search. The figure shows that the aerial-AF-relay consistently requires substantially higher power, even with less number of elements ($M<N$ in general), than the aerial-active-RIS under identical throughput conditions: gain produced by 25.48 and 27.19~dB in $d_{\mathcal{G}}=1000$ and $1200$~m, respectively. Note that beyond a certain distance in the figure, the AF-relay curve dramatically rises as the relay transmit power required to compensate for the multiplicative path loss becomes dominant~\cite{RISEE}. Furthermore, the pronounced gap from aerial-active-RIS and aerial-AF-relay stems from the fundamental architectural differences between the two systems: the aerial-AF-relay incurs an additional relay transmission stage, where the aerial-AF-relay is located at an intermediate position between the source and the receivers, to compensate for multiplicative path loss loss and amplified noise, together with a significantly larger RF-chain circuit power due to high-power components mentioned in~\eqref{afpower}, whereas the aerial-active-RIS benefits from cascaded source-RIS-destination reflection with amplification gain and relies only on low-power element-level amplification and control circuitry represented by $P_{\mathrm{E}}$.}
\subsection{Reliability of Placement and Amplification in Aerial-Active-RIS}
Fig.~\ref{fig_cubic} illustrates the variation of $\kappa_m$ with respect to the aerial-active-RIS height $(H)$ and the 2D distance of the UAV-BS from the source $(||\mathbf{w}_m||_2)$. It is observed that $\kappa_m$ is on the order of $10^{-2}$ when $||\mathbf{w}_m||_2$ is sufficiently large, implying that $\mathbf{q}_m^*=\kappa_m \mathbf{w}_m$ lies extremely close to the origin relative to $\mathbf{w}_m$. Moreover, $\kappa_m$ increases as $H$ grows. This occurs because a larger $H$ yields a greater $a$ and a smaller $\sqrt{-\frac{a}{3}}$ (with $a<0$) in~(\ref{xixixixi}), while the posterior term of $\sqrt{-\frac{a}{3}}$ in~(\ref{xixixixi}), approximated by $(-1+\frac{\epsilon}{3\sqrt{3}})$ for $|\epsilon|\ll 1$ in~(\ref{taylor}), remains negative. Consequently, increasing $H$ reduces $\sqrt{-\frac{a}{3}}$ and thereby enlarges $\kappa_m$. A similar analogy holds for decreasing $||\mathbf{w}_m||_2$, which likewise results in an increase of $\kappa_m$.

Fig.~\ref{fig_alpha}(a) and (b) show the behavior of $\alpha^{*2}$ according to $H$ and $||\mathbf{w}_m||_2$, and the optimal total power ($\mathrm{obj}$) with respect to $\alpha^2$ with minimum at $\alpha^{*2}$, respectively. As shown in (a), as $H$ and $||\mathbf{w}_m||_2$ increase, $\alpha^{*2}$ exhibits a rising trend, which is clear since increase in $H$ and $||\mathbf{w}_m||_2$ leads to farther link distance, which leads to stronger reflection compared to low $H$ and $||\mathbf{w}_m||_2$. This also can be found in~(\ref{optalpha}), where increase of $H$ and $||\mathbf{w}_m||_2$ leads to the increase of $\sum_{m\in\mathbb{M}}\Omega_0\left|\left|\pmb{\rho}_{\mathrm{RIS}}\right|\right|_2^2 \left|\left|\pmb{\rho}_{\mathrm{RIS}}-\pmb{\rho}_m \right|\right|_2^2 $. Moreover, we, in (b), can check the optimality of $\alpha^{*2}$, where the global minimum can be achieved in $\alpha^{*2}$ with minimum total power given in $\mathrm{obj}$ and~(\ref{alphaobj}). Consequently, if $\alpha^{*2}$ exceeds the threshold of 40~dB, we select $\alpha^*$ by $\alpha^{*2}=40$~dB as the optimum. From this point onward, we set $\alpha^2=\left(\alpha^{*2}10^{\frac{\Delta_0}{10}}\right)$~dB as a reference for comparing the non-optimality with respect to $\alpha^2 = \alpha^{*2}$.
\subsection{Comparison of Total Power under Various Conditions}
Fig.~\ref{actCenter}(a) illustrates the total power according to the distance of $\pmb{\rho}_{\mathcal{G}}$ from the source $\left(d_{\mathcal{G}}\right)$ with the implementation of active-RIS. Here, the performance gain by our proposed algorithm with $\alpha^2=\alpha^{*2}$ is approximately given by 32.20 and 30.17~dB for the distance 800 and 1200~m, respectively, which stably ensures almost 30~dB performance compared to total power of aerial-passive-RIS. It is also clear that $d_{\mathcal{G}}$ leads to the increase of total power, owing to the increase of path loss and probability to adopt the full-array scenario~\cite{HBRIS}. Moreover, as $d_{\mathcal{G}}$ increases, the power growth rate of the aerial-active-RIS becomes greater than that of the aerial-passive-RIS, which is shown by decrease in performance gain. It is because at longer distances, $\alpha^*$ becomes stronger, which also leads to an increase in both amplification power and dynamic noise which leads to the increase of total power.

We can also notice that the difference in total power between $\alpha^2=\alpha^{*2}$ and $(\alpha^{*2}10^{\frac{\Delta_0}{10}})$~dB with $\Delta_0=-5$ decreases from 3.79 to 1.52~dB for the distance of 800 to 1200~m, respectively, which implies the dependence on $\alpha$ near $\alpha^*$ decreases when $d_{\mathcal{G}}$ increases. We can analyze this phenomenon by taking the natural log of~(\ref{alphaobj}) by $p(\alpha)$ (i.e., $p(\alpha)=\ln(\mathrm{obj})$), and considering
\begin{equation}
\label{pert}
p(\alpha^*(1+\epsilon))-p(\alpha^*),\footnote{The power is measured in dBm, and we consider the variation of $\alpha$ in dB.}
\end{equation}
for small $\epsilon$.\footnote{Although $\epsilon$ is small, the product $\alpha$ can be significant due to the scale of $\alpha$, as shown in Fig.~\ref{fig_alpha}. Therefore, approximations such as using Taylor's first order expansion to estimate the difference are not reliable in this context. Instead, the total power difference in dBm should be computed directly.} It is evident that a smaller value of~(\ref{pert}) indicates a weaker sensitivity to $\alpha$ in the vicinity of $\alpha^*$, in terms of dB.
\begin{figure}[t]
	\begin{center}
		\includegraphics[width=0.85\columnwidth,keepaspectratio]%
		{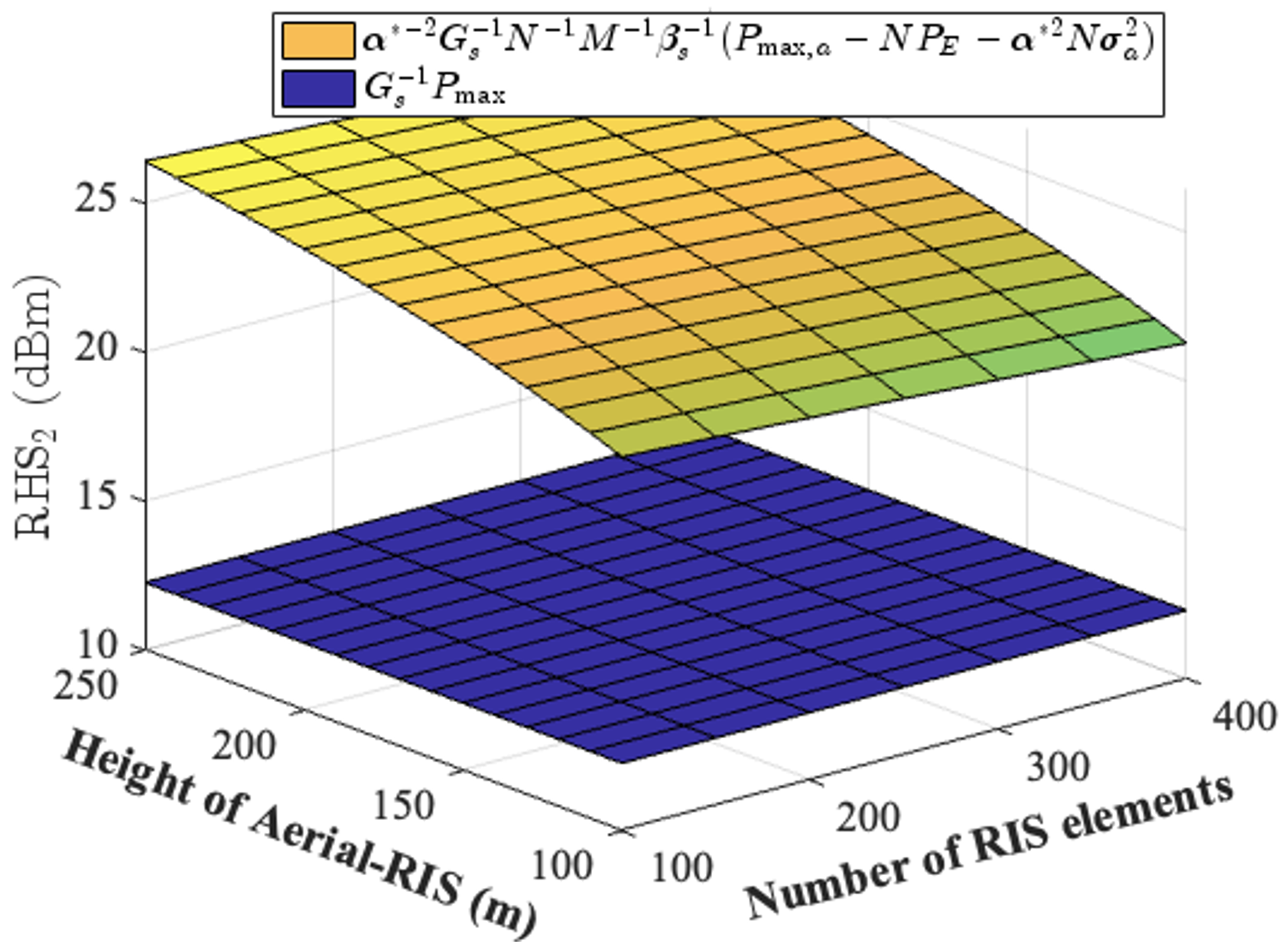}
		\caption{{Behavior of $\mathrm{RHS_2}$ in (\ref{objectACT}) with respect to the height of aerial-RIS and the number of active-RIS elements.}}
		\label{fig_feas}
	\end{center}
\end{figure}

Hence, by directly computing~(\ref{pert}) we can get~(\ref{mac4}). From~(\ref{mac4}), it is clear that as $d_{\mathcal{G}}$ increases, $||\pmb{\rho}_{\mathrm{RIS}} - \pmb{\rho}_m||_2^2$ also increases, causing the RHS of~(\ref{mac4}) to approach 0. Conversely, when $d_{\mathcal{G}}$ decreases, the argument inside the logarithm approaches $\frac{1}{2}\left(\left(1+\epsilon\right)^2 + \left(1+\epsilon\right)^{-2}\right)>1$. Therefore, we can conclude that an increase in $d_{\mathcal{G}}$ implies a reduced dependency on $\alpha$ in dB for the total power in dBm.
\begin{figure*}
\begin{equation}
\begin{aligned}
\label{mac4}
&p(\alpha^*(1+\epsilon))-p(\alpha^*)\\
&=\ln\left(\frac{\splitfrac{\left((1+\epsilon)^2+(1+\epsilon)^{-2}\right)\sqrt{\left( N\sigma_{\mathrm{a}}^2+\sum_{m\in\mathbb{M}}\Omega_0{\Omega_1}{\Omega_2}\right)\left(\sum_{m\in\mathbb{M}}\Omega_0\left|\left|\pmb{\rho}_{\mathrm{RIS}}\right|\right|_2^2 \left|\left|\pmb{\rho}_{\mathrm{RIS}}-\pmb{\rho}_m \right|\right|_2^2 \right)}\mathstrut}{\splitfrac{+\sum_{m\in\mathbb{M}}\left(\Omega_0\left|\left|\pmb{\rho}_{\mathrm{RIS}}\right|\right|_2^2{\Omega_2}+\Omega_0\left|\left|\pmb{\rho}_{\mathrm{RIS}}-\pmb{\rho}_m\right|\right|_2^2{\Omega_1}\right)+NP_{\mathrm{E}}}\mathstrut}}{\splitfrac{2\sqrt{\left( N\sigma_{\mathrm{a}}^2+\sum_{m\in\mathbb{M}}\Omega_0{\Omega_1}{\Omega_2}\right)\left(\sum_{m\in\mathbb{M}}\Omega_0\left|\left|\pmb{\rho}_{\mathrm{RIS}}\right|\right|_2^2 \left|\left|\pmb{\rho}_{\mathrm{RIS}}-\pmb{\rho}_m \right|\right|_2^2 \right)}\mathstrut}{\splitfrac{+\sum_{m\in\mathbb{M}}\left(\Omega_0\left|\left|\pmb{\rho}_{\mathrm{RIS}}\right|\right|_2^2{\Omega_2}+\Omega_0\left|\left|\pmb{\rho}_{\mathrm{RIS}}-\pmb{\rho}_m\right|\right|_2^2{\Omega_1}\right)+NP_{\mathrm{E}}}\mathstrut}}\right)
\end{aligned}
\end{equation}
\hrule
\end{figure*}


Fig.~\ref{actCenter}(b) illustrates the total power according to feasible $H$ with the guarantee of high LoS probability~\cite{a2gglobecom, Noh}. Clearly, by applying our algorithm with $\alpha^2 = \alpha^{*2}$, for increase of $H=160$ to $190$~m, we can reduce the total power by approximately 31.87 and 31.28~dB, respectively, compared to the passive-RIS-equipped benchmark algorithms. Similar to the scenario illustrated in Fig.~\ref{actCenter}, an increase in $H$ results in a greater distance between the UAV-BS and the aerial-RIS, which leads to, similar to Fig.~\ref{actCenter}, the decrease of performance gain. Furthermore, by applying the same reasoning to the RHS of~(\ref{mac4}), and noting that $||\pmb{\rho}_{\mathrm{RIS}}||_2^2, ||\pmb{\rho}_{\mathrm{RIS}} - \pmb{\rho}_m||_2^2 \sim \mathcal{O}(H^2)$, it follows that both the numerator and denominator of the RHS of~(\ref{mac4}) scale as $\mathcal{O}(H^2)$, which implies that the RHS of~(\ref{mac4}) scales as $\mathcal{O}(1)$ with respect to $H$. This indicates that the variation in total power, measured in dBm, with respect to $\alpha$ in dB remains nearly constant with increasing $H$, which is also shown in Fig.~\ref{actCenter}(b).

Fig.~\ref{actCenter}(c) shows the total power with respect to the number of active-RIS elements ($N$). {Note that in particular, the $N$-dependent circuit-related term appears as $(\alpha^2\sigma_{\mathrm{a}}^2+P_{\mathrm{E}})N$, whereas the remaining dominant terms in $\mathrm{obj}$ scale as $\mathcal{O}(1)$ with respect to $N$~\cite{HBRIS}. As a result, we can observe that the power increases in large scale of $N$. Nevertheless, the RIS is modeled as a ULA, to avoid illustrating unrealistic regimes with excessively large arrays, we therefore restrict $N$ to a feasible range, which is the left-side of Fig.~\ref{actCenter}(c).} As shown in the figure, although considering the reflection power $P_{\mathrm{R}}$ and power consumption by active hardware $NP_{\mathrm{E}}$, the proposed method with aerial-active-RIS greatly reduces the power, where the performance gain by our proposed algorithm with $\alpha^2=\alpha^{*2}$ is approximately given by 36.11 and 30.98~dB for $N=120$ and 300, respectively. 
In addition, we can notice that the difference of total power in dB between passive and active-RIS scenario is in first (small $N$) large enough, and gradually gets smaller in given range of the number of RIS elements. 
It is because when $N$ is small, $\mathrm{obj}$ is asymptotically $\mathrm{obj}\sim\sum_{m\in\mathbb{M}} P_m$, and by applying $N\rightarrow0$ to~(\ref{Constraint}), $P_m$ goes to $\alpha^{-2}P_{m,\mathrm{p}}$~\cite{HBRIS}. By combining those aspects, it is clear that the total power gain between aerial-active- and passive-RIS goes to $\alpha^2$, and by increasing $N$, the impact of the dynamic noise $\alpha^2 N\sigma_{\mathrm{a}}^2 $, the active-RIS-reflection power $\alpha^2 NM \beta_{\mathrm{s}} G_{\mathrm{s}}\sum_{m\in\mathbb{M}}P_m$, and the active-circuit-dissipated power $NP_{\mathrm{E}}$ is added to $\sum_{m\in\mathbb{M}}P_m$ and formulate $\mathrm{obj}$, affects to the increase of total power in combined way and the gain is becoming less than $\alpha^2$. 
Furthermore, both the numerator and the denominator of RHS of~(\ref{mac4}) scales as $\mathcal{O}\left(\sqrt{(N+N^2 g^{-1})g^{-1}}+Ng^{-1}+N\right),$ which also implies that~(\ref{mac4}) scales with $\mathcal{O}(1)$ with respect to $N$. This suggests that the total power perturbation from $\alpha^*$ to $\alpha^* + d$ in dBm stays constant with $N$, as illustrated in Fig.~\ref{actCenter}(c).

Fig.~\ref{fig_feas} shows the behavior of  $\mathrm{RHS}_2$ in~(\ref{objectACT}), 
which is the upper-bound of source transmit power $\sum_{m\in\mathbb{M}}P_m$. From the given figure, the expression
\begin{equation}
\label{expp}
\alpha^{*-2} G_{\mathrm{s}}^{-1}N^{-1}M^{-1}\beta_{\mathrm{s}}^{-1}\left(P_{\max,\mathrm{a}}-NP_{\mathrm{E}}-\alpha^{*2}  N \sigma_{\mathrm{a}}^{2}\right)
\end{equation}
in $\mathrm{RHS}_2$ decreases as $H$ and $N$ decrease, where the reduction is primarily due to the term $N^{-1}\beta_{\mathrm{s}}^{-1}$. Meanwhile, the term $G_{\mathrm{s}}^{-1}P_{\max}$ remains nearly constant, as it is almost a constant under the assumption that the antenna beamforming at the ground backhaul source is almost perfectly aligned with the aerial-active-RIS~\cite{beamJSAC, MS, beamgc}. {Moreover, the upper-bound is observed to be in range of 12$\sim$15~dBm.} Meanwhile, the total power in the simulations ($\mathrm{obj}$), presented in Fig.~\ref{actCenter}(a)$\sim$(c), almost surely remains below 10~dBm, which is strictly larger than the source transmit power:
\begin{equation}
\begin{aligned}
\label{strict}
\mathrm{obj}&\triangleq\left(1+\alpha^2 NM\beta_{\mathrm{s}} G_{\mathrm{s}}\right)\sum\limits_{m\in\mathbb{M}} P_{m}+\alpha^2 N\sigma_{\mathrm{a}}^2+ NP_{\mathrm{E}}\\
&>\sum_{m\in\mathbb{M}}P_m.
\end{aligned}
\end{equation}
Hence, we can observe that it leads to the fact that $\sum_{m\in\mathbb{M}}P_m$ also stays below 10~dBm, which does not exceeds the $\mathrm{RHS}_2$ that are computed in Fig.~\ref{fig_feas}. {This equivalently indicates that the proposed system remains within the feasible region defined by $\mathrm{RHS}_2$ in~(\ref{objectACT}).}

\section{Conclusion}
In this paper, we proposed a novel aerial-active-RIS-assisted backhaul architecture to enable energy-efficient full 3D coverage for UAV-BS backhaul networks in 6G. We derived the minimum total power required for backhauling UAV-BSs under target data rate constraints and showed that equal amplification gain is an effective strategy for maximizing energy-efficiency. By employing a practical aerial-active-RIS signal model and accounting for active-RIS-induced dynamic noise, we optimized the placement, array configuration, amplification gain, and phase of the aerial-active-RIS. Simulation results validated the effectiveness of the proposed method, demonstrating significant energy-efficiency improvements over benchmarks and highlighting its strong potential for delivering reliable and scalable 3D backhaul coverage in 6G.


	\section*{Appendix A}
\section*{Proof of Theorem~\ref{cubiceq}}
Let $\mathbf{\bar{q}}_m$ be the orthogonal projection of $\mathbf{q}_m$ onto the line segment connecting the source and $\mathbf{w}_m$. Thereafter, the following holds:
\begin{equation}
\begin{split}
\label{projq}
\begin{cases}
\left\|\mathbf{q}_m - \mathbf{w}_m\right\|_2^2 = \left\|\mathbf{q}_m - \mathbf{\bar{q}}_m\right\|_2^2 + \left\|\mathbf{\bar{q}}_m - \mathbf{w}_m\right\|_2^2\\
\left\|\mathbf{q}_m \right\|_2^2 = \left\|\mathbf{q}_m - \mathbf{\bar{q}}_m\right\|_2^2 +\left\|\mathbf{\bar{q}}_m \right\|_2^2.
\end{cases}
\end{split}
\end{equation}  
To minimize the objective in~(\ref{num}), it is necessary to minimize the left-hand side in~(\ref{projq}). Consequently, $\mathbf{q}_m$ should satisfy:
\begin{equation}
\label{projcond}
\left\|\mathbf{q}_m - \mathbf{\bar{q}}_m\right\|_2^2 =0  \leftrightarrow  \mathbf{q}_m = \kappa_m \mathbf{w}_m ~(\kappa_m>0),
\end{equation}
			\begin{figure}[t]
	\begin{center}
		\includegraphics[width=0.9\columnwidth,keepaspectratio]%
		{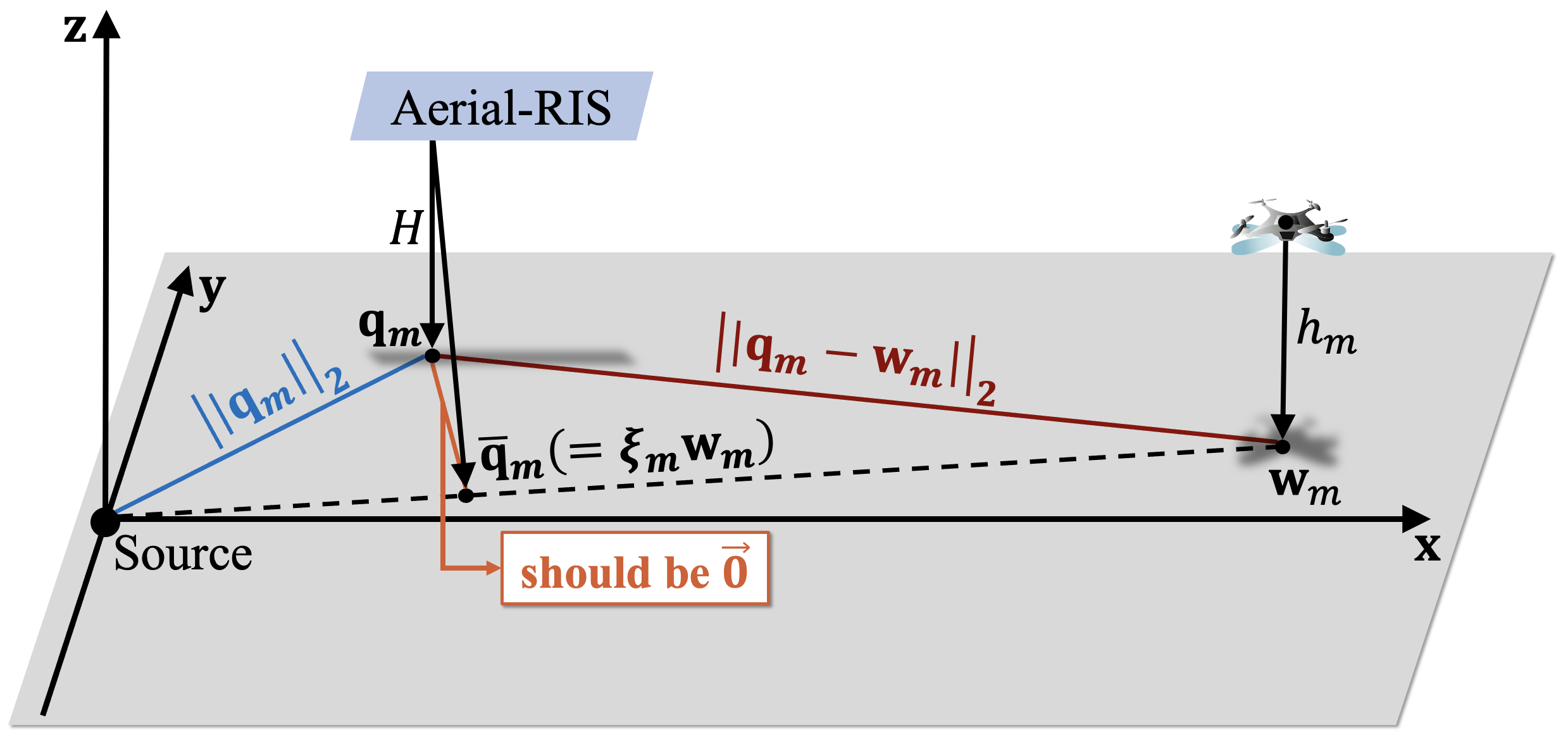}
		\caption{Illustration of~(\ref{projcond}) in Theorem~\ref{cubiceq}: orthogonal projection of $\mathbf{q}_m$.}
		\label{figth2}
	\end{center}
\end{figure}
which is clarified in Fig.~\ref{figth2}. Hence, by substituting~(\ref{projcond}) we can denote the objective of~(\ref{num}) by $g\left(\kappa_m\right)$, that is:
	\begin{equation}
		\label{gkappa}
		g\left(\kappa_m\right)\triangleq \left|\left|\mathbf{w}_m\right|\right|_2^4 \left(\kappa_m^2 + \zeta_1^2 +\bar{\Omega}_1 \right) \left(\left(1-\kappa_m\right)^2 + \zeta_2^2 +\bar{\Omega}_2 \right),
	\end{equation}
where
\begin{equation}
\label{zeraexp}
\zeta_1=\frac{H}{\left|\left|\mathbf{w}_m\right|\right|_2}, \zeta_2 = \frac{\left|H-h_m\right|}{\left|\left|\mathbf{w}_m\right|\right|_2}, \bar{\Omega}_i= ||\mathbf{w}_m||_2^{-2}\tilde{\Omega}_i~(i=1,2).
\end{equation}
To find the minimum of $g$ at $\kappa_m>0$, we have to solve $g^{\prime}(\kappa_m)=0$. The discriminant $\Delta$ of the cubic equation $g'\left(\kappa_m\right)=0$ is given by $\Delta = \left(\frac{a}{3}\right)^3 + \left(\frac{b}{2}\right)^2$~\cite{Lovett}, where $a=\frac{1}{2} \left(\zeta_1^2 + \zeta_2^2 +\bar{\Omega}_1+\bar{\Omega}_2 \right) - \frac{1}{4}, b=\frac{1}{4} \left(\zeta_2^2 - \zeta_1^2 +\bar{\Omega}_2 - \bar{\Omega}_1\right).$ For the range of $a$ and $b$, we can assume the followings:
\begin{enumerate}
\item Since we assume that $d_{\mathcal{G}}$ is sufficiently large, $\left|\left|\mathbf{w}_m\right|\right|_2$ follows a similar scale. In this paper, we will assume the scale with $10^3$~m, which is also reflected in the simulation in Section IV and Table~\ref{SimPar}.
\item $\beta_0\approx-43.3$~dB for sub-6~GHz backhaul (gets larger for higher-frequency applications in 6G~\cite{dsRIS, trimimo,trimimo2}).
\item Feasible $\alpha_{\max}^2$ is given by less than 40~dB~\cite{aris1, aris4}.
\item $G_{\mathrm{s}}\le G_{\max}=8$~dB in Table~\ref{SimPar}.
\item We assume that $M$ and $N$ has a scale of approximately a few or less than ten and few hundred, respectively: $M=16, N\in[100, 400]$~\cite{RISEE, HBRIS, aris1}
\item $\sigma_{\mathrm{a}}^2$ and $\sigma^2$ has the similar scale~\cite{aris4, aris5}.
\end{enumerate}
Therefore, we can assume that
\begin{equation}
\begin{aligned}
\label{approxKK}
&\zeta_i\ll1~(i=1,2)\\
&\bar{\Omega}_1=\alpha^2 NM \beta_0 G_{\mathrm{s}} ||\mathbf{w}||_2^{-2}\ll1,~\bar{\Omega}_2=\alpha^2 \frac{\sigma_{\mathrm{a}}^2}{\sigma^2} N \beta_0 ||\mathbf{w}_m||_2^{-2}\ll1,
\end{aligned}
\end{equation}
which also leads to $a<0\left(\approx-\frac{1}{4}\right)$ and $|b|\ll1$. 
Hence, we can deduce $\Delta<0$, which leads to three real solutions $\left\{\kappa_{m,k}\right\}_{k=0}^2$ of $g'\left(\kappa_m\right)=0$~\cite{Lovett}:
	\begin{equation}
		\begin{split}
		\label{solcubicACT}
		\kappa_{m,k}=\frac{1}{2}+2\sqrt{-\frac{a}{3}}\cos \left( \frac{1}{3} \cos^{-1} \left( \frac{3b}{2a} \sqrt {-\frac{3}{a}} \right)-\frac{2}{3}\pi k \right)&\\
		\left(k=0,1,2\right)&.
	\end{split}
	\end{equation}
Since $a \approx -\tfrac{1}{4}$ and $|b| \ll 1$, we define $\tfrac{3b}{2a}\sqrt{-\tfrac{3}{a}} = \epsilon$ with $|\epsilon| \ll 1$. By successively applying the first-order Taylor approximation to~(\ref{solcubicACT}), the expression reduces to
	\begin{equation}
	\begin{split}
	\label{taylor}
	\begin{cases}
	\kappa_{m,0} \approx \frac{1}{2}+\sqrt{-a} \left(1+ \frac{\epsilon}{3\sqrt{3}} \right)\\
	\kappa_{m,1} \approx \frac{1}{2}+\sqrt{-a} \left(- \frac{\epsilon}{3\sqrt{3}} \right)\\
	\kappa_{m,2} \approx \frac{1}{2}+\sqrt{-a} \left(-1+ \frac{\epsilon}{3\sqrt{3}} \right).
	\end{cases}
	\end{split}
	\end{equation}
	By substituting $a=\frac{1}{2} \left(\zeta_1^2 + \zeta_2^2 +\bar{\Omega}_1+\bar{\Omega}_2 \right) - \frac{1}{4}$,~(\ref{taylor}) becomes
		\begin{equation}
	\begin{split}
	\label{taylor2}
	\begin{cases}
	\kappa_{m,0} \approx \frac{1}{2}-\sqrt{\frac{1}{4}-\frac{1}{2} \left(\zeta_1^2 + \zeta_2^2 +\bar{\Omega}_1+\bar{\Omega}_2\right)} \left(-1- \frac{\epsilon}{3\sqrt{3}} \right)\\
	\kappa_{m,1} \approx \frac{1}{2}-\sqrt{\frac{1}{4}-\frac{1}{2} \left(\zeta_1^2 + \zeta_2^2+\bar{\Omega}_1+\bar{\Omega}_2 \right)} \left( \frac{\epsilon}{3\sqrt{3}} \right)\\
	\kappa_{m,2} \approx \frac{1}{2}-\sqrt{\frac{1}{4}-\frac{1}{2} \left(\zeta_1^2 + \zeta_2^2 +\bar{\Omega}_1+\bar{\Omega}_2\right)} \left(1- \frac{\epsilon}{3\sqrt{3}} \right).\\
	\end{cases}
	\end{split}
	\end{equation}	
	From~(\ref{taylor2}), we can derive the following outcomes:
	 \begin{enumerate}
 \item It is evident that $\kappa_{m,0}>\kappa_{m,1}>\kappa_{m,2}~(\because |\epsilon|\ll1)$.
 \item By the properties of the quartic equation~\cite{Lovett}, $g$ has two local minimum: $\kappa_{m,0}$ and $\kappa_{m,2}$, with one of them being the global minimum.
 \item Using $|\epsilon|\ll1$, we can deduce the following for $\kappa_{m,2}$:
 \begin{equation}
 \begin{aligned}
 \label{taylor3}
 \kappa_{m,2}&\approx \frac{1}{2}-\sqrt{\frac{1}{4} -\frac{1}{2} \left(\zeta_1^2 + \zeta_2^2 +\bar{\Omega}_1+\bar{\Omega}_2 \right)} \left(1\right)\\
 &=\frac{\frac{1}{2} \left(\zeta_1^2 + \zeta_2^2 +\bar{\Omega}_1+\bar{\Omega}_2 \right)}{\sqrt{\frac{1}{4} + \frac{1}{2} \left(\zeta_1^2 + \zeta_2^2 +\bar{\Omega}_1+\bar{\Omega}_2 \right)}}>0.
 \end{aligned}
 \end{equation}
Furthermore, as $\zeta_1$ and $\zeta_2$ are sufficiently small (owing to the sufficiently large $||\mathbf{w}_m||_2$), it follows from~(\ref{taylor3}) that $\kappa_{m,2}$ should be close to the origin. 
 \end{enumerate}
Hence, we should select $\kappa_m$ as
	\begin{equation}
	\label{cubicres}
	\begin{aligned}
	\kappa_m&\triangleq\kappa_{m,2}\\
	&=\frac{1}{2}+2\sqrt{-\frac{a}{3}}\cos \left( \frac{1}{3} \cos^{-1} \left( \frac{3b}{2a} \sqrt {-\frac{3}{a}} \right)-\frac{4}{3}\pi  \right).
	\end{aligned}
	\end{equation}
	By determining $\mathbf{q}_m^*=\kappa_m\mathbf{w}_m$ accordingly, the theorem follows.~$\blacksquare$
		
	\section*{Appendix B}
\section*{Proof of Theorem~\ref{thmalpha}}
For given $\mathbf{q}^*, \{\bar{N}\}, \pmb{\bar\rho}~(\mathrm{or}~\{\pmb{\bar\rho}_i^*\}_{i=1}^L), \pmb\Theta$, $\mathrm{obj}$ in~(\ref{objectACT}) becomes
 \begin{equation}
 \begin{aligned}
 \label{alphaobj}
 &\left(1+\alpha^2 NM\beta_{\mathrm{s}} G_{\mathrm{s}}\right)\sum\limits_{m\in\mathbb{M}} P_{m}+\alpha^2 N\sigma_{\mathrm{a}}^2+ NP_{\mathrm{E}}(\triangleq\mathrm{obj})\\
&=\left( N\sigma_{\mathrm{a}}^2+\sum_{m\in\mathbb{M}}\Omega_0{\Omega_1}{\Omega_2}\right)\alpha^2\\
&~~~+\left(\sum_{m\in\mathbb{M}}\Omega_0\left|\left|\pmb{\rho}_{\mathrm{RIS}}\right|\right|_2^2 \left|\left|\pmb{\rho}_{\mathrm{RIS}}-\pmb{\rho}_m \right|\right|_2^2 \right)\alpha^{-2}\\
&~~~+\sum_{m\in\mathbb{M}}\left(\Omega_0\left|\left|\pmb{\rho}_{\mathrm{RIS}}\right|\right|_2^2{\Omega_2}+\Omega_0\left|\left|\pmb{\rho}_{\mathrm{RIS}}-\pmb{\rho}_m\right|\right|_2^2{\Omega_1}\right)+NP_{\mathrm{E}},
 \end{aligned}
 \end{equation}
Hence, we can derive the optimal $\alpha$ by applying the Arithmetic-Geometric Mean inequality, which becomes~\eqref{arge}.
\begin{figure*}
\begin{equation}
\begin{aligned}
\label{arge}
&\left(1+\alpha^2 NM\beta_{\mathrm{s}} G_{\mathrm{s}}\right)\sum\limits_{m\in\mathbb{M}} P_{m}+\alpha^2 N\sigma_{\mathrm{a}}^2+ NP_{\mathrm{E}}\\
&\ge 2\sqrt{\left( N\sigma_{\mathrm{a}}^2+\sum_{m\in\mathbb{M}}\Omega_0{\Omega_1}{\Omega_2}\right)\left(\sum_{m\in\mathbb{M}}\Omega_0\left|\left|\pmb{\rho}_{\mathrm{RIS}}\right|\right|_2^2 \left|\left|\pmb{\rho}_{\mathrm{RIS}}-\pmb{\rho}_m \right|\right|_2^2 \right)}+\sum_{m\in\mathbb{M}}\left(\Omega_0\left|\left|\pmb{\rho}_{\mathrm{RIS}}\right|\right|_2^2{\Omega_2}+\Omega_0\left|\left|\pmb{\rho}_{\mathrm{RIS}}-\pmb{\rho}_m\right|\right|_2^2{\Omega_1}\right)+NP_{\mathrm{E}}.
\end{aligned}
\end{equation}
\hrule
\end{figure*}
Therefore, by the equality condition of~(\ref{arge}), 
we can deduce the optimal $\alpha^*$ as
\begin{equation}
\label{optalpha}
\alpha^*=\min\left\{\sqrt[4]{\frac{\sum_{m\in\mathbb{M}}\Omega_0\left|\left|\pmb{\rho}_{\mathrm{RIS}}\right|\right|_2^2 \left|\left|\pmb{\rho}_{\mathrm{RIS}}-\pmb{\rho}_m \right|\right|_2^2 }{N\sigma_{\mathrm{a}}^2+\sum_{m\in\mathbb{M}}\Omega_0{\Omega_1}{\Omega_2}}}, \alpha_{\max}\right\},
\end{equation}
and the theorem follows.~$\blacksquare$	

		\bibliographystyle{IEEEtran}
		\bibliography{IEEEexample}	

@string{CONF_GLOBECOMWH		= "Proc. IEEE Glob. Commun. Conf. (GLOBECOM) Workshops"}

@string{IEEE_J_COMSURVTUT	= "{IEEE} Commun. Surveys Tuts."}

@string{IEEE_J_JSAC			= "{IEEE} J. Sel. Areas Commun."}

@string{IEEE_J_COM			= "{IEEE} Trans. Commun."}

@string{IEEE_J_TSP			= "{IEEE} Trans. Signal Process."}

@string{IEEE_J_TVT			= "{IEEE} Trans. Veh. Technol."}

@string{IEEE_J_WCOM			= "{IEEE} Trans. Wireless Commun."}

@string{IEEE_M_COMM			= "{IEEE} Commun. Mag."}

@string{IEEE_L_COML			= "{IEEE} Commun. Lett."}

@string{IEEE_L_WCOML		= "{IEEE} Wireless Commun. Lett."}

@misc{ITU525,
	author="{ITU-R P.525-2}",
	title="Calculation of free-space attenuation",
	year="1994"
}

@inproceedings{a2gglobecom,
  title={Modeling air-to-ground path loss for low altitude platforms in urban environments},
  author={A. Al-Hourani and others},
  booktitle=CONF_GLOBECOMWH,
  pages={2898-2904},
  month=dec,
  year={2014}
}

@BOOK{beamwidth,
title ={{Antenna Theory: Analysis and Design}},
author ={C. A. Balanis},
year ={2016},
PUBLISHER ={New York, NY, USA: Wiley}
}

@ARTICLE{battery,
	author={Y. {Zeng} and others},
	journal=IEEE_J_WCOM,
	title={Energy Minimization for Wireless Communication With Rotary-Wing {UAV}},
	year={2019},
	volume={18},
	number={4},
	pages={2329-2345},
	month={April},
}

@ARTICLE{LingRIS,
  author={Dai, Linglong and others},
  	journal={IEEE Access}, 
	title={Reconfigurable Intelligent Surface-Based Wireless Communications: Antenna Design, Prototyping, and Experimental Results}, 
	year={2020},
	volume={8},
	number={},
	pages={45913-45923},
}

@ARTICLE{aerial3d,
  author={Liu, Jianghui and Zhang, Hongtao},
  journal=IEEE_J_TVT, 
  title={Throughput Optimization in Aerial {RIS}-Assisted Networks With {3D} Imperfect Reflection}, 
  year={2025},
  volume={74},
  number={7},
  pages={10510-10523},
  month=jul,
  }

@INPROCEEDINGS{UAVRIS,
  author={H. {Lu} and others},
	booktitle={Proc. IEEE  Int. Conf. on Comm. (ICC) Workshops}, 
	title={Enabling Panoramic Full-Angle Reflection Via Aerial Intelligent Reflecting Surface}, 
	year={2020},
	volume={},
	number={},
	pages={1-6},
}

@ARTICLE{RISEE,
  author={Huang, Chongwen and others},
 	journal=IEEE_J_WCOM,
	title={Reconfigurable Intelligent Surfaces for Energy Efficiency in Wireless Communication}, 
	year={2019},
	volume={18},
	number={8},
	month=Aug,
	pages={4157-4170},
}

@ARTICLE{Noh,
  author={S.-C. {Noh} and others},
	journal=IEEE_L_WCOML,
	title={Energy-Efficient Deployment of Multiple {UAV}s Using Ellipse Clustering to Establish Base Stations}, 
	year={2020},
	volume={9},
	number={8},
	month=Aug,
	pages={1155-1159},
}

@ARTICLE{MS,
  author={M. S. {Sim} and others},
	journal={IEEE Access}, 
	title={Deep Learning-Based mm{W}ave Beam Selection for {5G NR/6G} With Sub-6 {GHz} Channel Information: Algorithms and Prototype Validation}, 
	year={2020},
	volume={8},
	number={},
	pages={51634-51646},
}

@Book{MVO,
title = {Nonlinear Multiobjective Optimization},
author = {K. Miettinen},
year={1998},
PUBLISHER={Dordrecht, Netherlands: Kluwer Academic Publishers}
}

@BOOK{HPBW,
	title ={Phased Array Antennas},
	author ={R. C. Hansen},
	year ={2009},
	PUBLISHER ={Hoboken, NJ, USA: John Wiley \& Sons},
}

@BOOK{NR,
	title ={Study on Channel Model for Frequencies From 0.5 to 100 {GHz}},
	year ={2020},
	month=jan,
	PUBLISHER ={document 3GPP TR 38.901},
}

@BOOK{38.801,
	title ={Study on new radio access technology: Radio access architecture and interfaces},
	year ={2016},
	month=jan,
	PUBLISHER ={document 3GPP TR 38.801},
}

@article{FT,
	Author = {E. Weiszfeld and F. Plastria},
	Journal = {Ann. Oper. Res.},
	Number = {1},
	Pages = {7--41},
	Title = {On the point for which the sum of the distances to $n$ given points is minimum},
	Volume = {167},
	Year = {2009},
	month=Mar,
}

@book{boyd,
	title={Convex Optimization},
	author={Boyd, S and Vandenberghe, L},
	year={2004},
	publisher={Cambridge, UK: Cambridge Univ. Press},
}

@book{Lovett,
  title={Abstract Algebra: Structures and Applications},
  author={S. Lovett},
  year={2015},
  publisher={Boca Raton, FL, USA: CRC Press}
}

@article{WF2,
  title={Weiszfeld’s method: Old and new results},
  author={A. Beck and S. Sabach},
  journal={J. Optim. Theory Appl.},
  volume={164},
  number={1},
  pages={1--40},
  year={2015},
  month=may,
}

@ARTICLE{DF,
  author={E. Bj\"ornson and others},
   journal={IEEE Wireless Commun. Lett.}, 
  title={Intelligent Reflecting Surface Versus Decode-and-Forward: How Large Surfaces are Needed to Beat Relaying?}, 
  year={2020},
  volume={9},
  number={2},
  pages={244-248},
  month=feb,
  }

@ARTICLE{vsrelay,
  author={M. {Di Renzo} and others},
  journal={IEEE Open J. Commun. Soc.}, 
  title={Reconfigurable Intelligent Surfaces vs. Relaying: Differences, Similarities, and Performance Comparison}, 
  year={2020},
  volume={1},
  number={},
  pages={798-807},
  month=jun,
  }

@ARTICLE{RIST,
  author={Wu, Qingqing and others},
   journal={IEEE Trans. Commun.}, 
  title={Intelligent Reflecting Surface-Aided Wireless Communications: A Tutorial}, 
  year={2021},
  volume={69},
  number={5},
  pages={3313-3351},
  month=may,
  }

@ARTICLE{meta,
  author={Liaskos, Christos and others},
   journal={IEEE Commun. Mag.}, 
  title={A New Wireless Communication Paradigm through Software-Controlled Metasurfaces}, 
  year={2018},
  volume={56},
  number={9},
  pages={162-169},
  month=sep,
  }

@ARTICLE{mozatut,
  author={Mozaffari, Mohammad and others},
   journal=IEEE_J_COMSURVTUT,
  title={A Tutorial on {UAVs} for Wireless Networks: Applications, Challenges, and Open Problems}, 
  year={2019},
  volume={21},
  number={3},
  pages={2334-2360},
  month={Third quarter},
  }

@ARTICLE{UPARIS,
  author={H. Lu and others},
  journal=IEEE_J_WCOM,
  title={Aerial Intelligent Reflecting Surface: Joint Placement and Passive Beamforming Design with {3D} Beam Flattening}, 
  year={2021},
  volume={20},
  number={7},
  pages={4128-4143},
  month=jul,
  }

@ARTICLE{CE1,
  author={Wei, Li and others},
  journal=IEEE_J_COM,
  title={Channel Estimation for {RIS}-Empowered Multi-User {MISO} Wireless Communications}, 
  year={2021},
  volume={69},
  number={6},
  pages={4144-4157},
  month=jun,
  }

@ARTICLE{CE2,
  author={Wei, Xiuhong and others},
  journal=IEEE_L_COML,
  title={Channel Estimation for {RIS} Assisted Wireless Communications—{Part} {II}: An Improved Solution Based on Double-Structured Sparsity}, 
  year={2021},
  volume={25},
  number={5},
  pages={1403-1407},
  month=may,
  }

@article{HBFSO,
	title={Free-Space Optical Communications for 6{G} Wireless Networks: Challenges, Opportunities, and Prototype Validation},
	author={Hong-Bae Jeon and others},
	journal={IEEE Commun. Mag.},
	year={2023},
	  volume={61},
  number={4},
  pages={116-121},
  month=apr,
}

@ARTICLE{aris1,
  journal=IEEE_J_COM,	
	title={Active {RIS} vs. Passive {RIS}: Which Will Prevail in {6G}?},
	author={Z. Zhang and others},
  year={2023},
  volume={71},
  number={3},
  pages={1707-1725},
  month=mar,
  }

@ARTICLE{aris2,
  author={Zhu, Qi and others},
  journal=IEEE_J_COM,	
  title={Joint Beamforming Designs for Active Reconfigurable Intelligent Surface: A Sub-Connected Array Architecture}, 
  year={2022},
  volume={70},
  number={11},
  pages={7628-7643},
  month=nov,
  }

@ARTICLE{aris3,
  author={Liu, Kunzan and others},
  journal=IEEE_L_COML,
  title={Active Reconfigurable Intelligent Surface: Fully-Connected or Sub-Connected?}, 
  year={2022},
  volume={26},
  number={1},
  pages={167-171},
  month=jan,
  }

@INPROCEEDINGS{HBICTC,
  author={Jeon, Hong-Bae and Chae, Chan-Byoung},
  booktitle={Proc. IEEE Int. Conf. ICT Converg. (ICTC)}, 
  title={Energy-Efficient Aerial-{RIS} Deployment for {6G}}, 
  year={2022},
  volume={},
  number={},
  pages={199-201},
  }

@ARTICLE{aris4,
  author={Zhi, Kangda and others},
  journal=IEEE_L_COML,
  title={Active {RIS} Versus Passive {RIS}: Which is Superior With the Same Power Budget?}, 
  year={2022},
  volume={26},
  number={5},
  pages={1150-1154},
  month=may,
  }

@ARTICLE{aris5,
  author={Long, Ruizhe and others},
    journal={IEEE Trans. Wireless Commun.}, 
  title={Active Reconfigurable Intelligent Surface-Aided Wireless Communications}, 
  year={2021},
  volume={20},
  number={8},
  pages={4962-4975},
  month=aug,
  }

@ARTICLE{HBRIS,
  author={Jeon, Hong-Bae and others},
    journal={IEEE Trans. Wireless Commun.}, 
  title={An Energy-efficient Aerial Backhaul System with Reconfigurable Intelligent Surface}, 
  year={2022},
  volume={21},
  number={8},
  pages={6478-6494},
  month=aug,
  }

@ARTICLE{Doh,
  author={Do, Heedong and Lee, Namyoon},
    journal={IEEE Trans. Wireless Commun.}, 
  title={Finding Globally Optimal Configuration of Active {RIS} in Linear Time}, 
  year={2024},
  volume={23},
  number={12},
  pages={18142-18153},
  month=dec,
  }

@ARTICLE{wongjsac,
  author={Niu, Hehao and others},
  journal=IEEE_J_JSAC,
  title={Active {RIS} Assisted Rate-Splitting Multiple Access Network: Spectral and Energy Efficiency Tradeoff}, 
  year={2023},
  volume={41},
  number={5},
  pages={1452-1467},
  month=may,
  }

@ARTICLE{CE3,
  author={Chen, Jie and others},
    journal={IEEE Trans. Wireless Commun.}, 
  title={Channel Estimation for Reconfigurable Intelligent Surface Aided Multi-User {mmWave} {MIMO} Systems}, 
  year={2023},
  volume={22},
  number={10},
  pages={6853-6869},
  month=oct,
  }

@ARTICLE{CE4,
  author={Zhou, Gui and others},
    journal={IEEE Trans. Wireless Commun.}, 
  title={Individual Channel Estimation for {RIS}-Aided Communication Systems—A General Framework}, 
  year={2024},
  volume={23},
  number={9},
  pages={12038-12053},
  month=sep,
  }

@ARTICLE{dsRIS,
  author={Jun, Dongsoo and others},
  journal={IEEE Commun. Stand. Mag.}, 
  title={Reconfigurable Intelligence Surface with Potential Tunable Meta-Devices for {6G}: Design and System-Level Evaluation}, 
  year={2024},
  volume={8},
  number={4},
  pages={32-39},
  month=apr,
  }

@ARTICLE{trimimo,
  author={Castellanos, Miguel Rodrigo and others},
  journal={IEEE Trans. Commun.}, 
	title={Embracing Reconfigurable Antennas in the Tri-hybrid {MIMO} Architecture for {6G} and Beyond},
  year={2026},
  volume={74},
  number={1},
  pages={381-401},
  }

@ARTICLE{trimimo2,
  author={{Heath, Jr.}, Robert W. and others},
  journal={IEEE Wireless Commun.}, 
	title={The Tri-hybrid {MIMO} Architecture},
  year={2026},
  volume={},
  number={},
  pages={1-7},
  }

@ARTICLE{Din,
author = {Alessio Zappone and Eduard Jorswieck},
  journal={Found. Trends Commun. Inf. Theory}, 
  title={Energy Efficiency in Wireless Networks via Fractional Programming Theory}, 
  year={2015},
  volume={11},
  number={3-4},
  pages={185-396},
  }

@ARTICLE{beamJSAC,
  author={Liang, Kai and others},
  journal=IEEE_J_JSAC,
  title={A Data and Model-Driven Deep Learning Approach to Robust Downlink Beamforming Optimization}, 
  year={2024},
  volume={42},
  number={11},
  pages={3278-3292},
  month=nov,
  }

@ARTICLE{yhFD,
  author={Kim, Yonghwi and others},
  journal={Proc. IEEE}, 
  title={A State-of-the-Art Survey on Full-Duplex Network Design}, 
  year={2024},
  volume={112},
  number={5},
  pages={463-486},
  month=may,
  }

@ARTICLE{hjcoop,
  author={Moon, Hyung-Joo and Chae, Chan-Byoung},
  journal=IEEE_J_JSAC,
  title={Cooperative Ground-Satellite Scheduling and Power Allocation for Urban Air Mobility Networks}, 
  year={2025},
  volume={43},
  number={1},
  pages={218-233},
  month=jan,
  }

@ARTICLE{hjgen,
  author={Moon, Hyung-Joo and others},
  journal={IEEE Trans. Wireless Commun.},
  title={A Generalized Pointing Error Model for {FSO} Links with Fixed-Wing {UAVs} for 6{G}: Analysis and Trajectory Optimization}, 
  year={2025},
  volume={24},
  number={7},
  pages={5723-5737},
  month=jul,
  }

@ARTICLE{smidaFD,
  author={Smida, Besma and others},
  journal=IEEE_J_JSAC,
  title={Full-Duplex Wireless for 6{G}: Progress Brings New Opportunities and Challenges}, 
  year={2023},
  volume={41},
  number={9},
  pages={2729-2750},
  month={Sep.},
  }

@ARTICLE{ntn6g,
  author={Giordani, Marco and Zorzi, Michele},
  journal={IEEE Netw.}, 
  title={Non-Terrestrial Networks in the 6{G} Era: Challenges and Opportunities}, 
  year={2021},
  volume={35},
  number={2},
  pages={244-251},
  month={Mar./Apr.},
  }

@article{bjorn6G,
	title={Towards 6{G MIMO}: Massive Spatial Multiplexing, Dense Arrays, and Interplay Between Electromagnetics and Processing},
	author={E. Bj\"ornson and others},
	journal={arXiv:2401.02844v1},
	year={2024}
}

@ARTICLE{actpower1,
  author={Bousquet, Jean-François and others},
  journal={IEEE Trans. Circuits Syst. I, Reg. Papers,}, 
  title={A 4-{GHz} Active Scatterer in 130-nm {CMOS} for Phase Sweep Amplify-and-Forward}, 
  year={2012},
  volume={59},
  number={3},
  pages={529-540},
  month=mar,
  }

@ARTICLE{actpower2,
  author={Amato, Francesco and Peterson, Christopher W. and Degnan, Brian P. and Durgin, Gregory D.},
  journal={IEEE J. Radio Freq. Identificat.}, 
  title={Tunneling {RFID} Tags for Long-Range and Low-Power Microwave Applications}, 
  year={2018},
  volume={2},
  number={2},
  pages={93-103},
  month=jun,
  }

@ARTICLE{risspm,
  author={Basar, Ertugrul and Poor, H. Vincent},
  journal={IEEE Signal Process. Mag.}, 
  title={Present and Future of Reconfigurable Intelligent Surface-Empowered Communications [Perspectives]}, 
  year={2021},
  volume={38},
  number={6},
  pages={146-152},
  month=nov,
  }

@ARTICLE{aristut,
  author={Ahmed, Manzoor and others},
  journal=IEEE_J_COMSURVTUT, 
  title={Active Reconfigurable Intelligent Surfaces: Expanding the Frontiers of Wireless Communication-A Survey}, 
  year={2025},
  volume={27},
  number={2},
  pages={839-869},
  month=apr,
  }

@ARTICLE{aris7,
  author={Zhou, Gui and others},
  journal={IEEE Trans. Wireless Commun.}, 
  title={A Framework for Transmission Design for Active {RIS}-Aided Communication With Partial {CSI}}, 
  year={2024},
  volume={23},
  number={1},
  pages={305-320},
  month={Jan.},
  }

@ARTICLE{alexRIS,
  author={Alexandropoulos, George C. and others},
  journal=IEEE_M_COMM,
  title={Reconfigurable Intelligent Surfaces for Rich Scattering Wireless Communications: Recent Experiments, Challenges, and Opportunities}, 
  year={2021},
  volume={59},
  number={6},
  pages={28-34},
  month=jun,
  }

@ARTICLE{uavdep1,
  author={Zhang, Yongqiang and others},
  journal=IEEE_J_WCOM, 
  title={Deployment Optimization of Tethered Drone-Assisted Integrated Access and Backhaul Networks}, 
  year={2024},
  volume={23},
  number={4},
  pages={2668-2680},
  month={Apr.},
  }

@ARTICLE{uavdep2,
  author={Diaz-Vilor, Carles and others},
  journal=IEEE_J_WCOM, 
  title={Cell-Free {UAV} Networks: Asymptotic Analysis and Deployment Optimization}, 
  year={2023},
  volume={22},
  number={5},
  pages={3055-3070},
  month={May},
  }

@ARTICLE{aris8,
  author={Yang, Jinho and others},
  journal=IEEE_J_TVT,
  title={Robust Transmission Design for Active {RIS}-Aided Systems}, 
  year={2025},
  volume={},
  number={},
  pages={1-6},
}

@article{RISRe,
	title={Multi-Antenna Relaying and Reconfigurable Intelligent Surfaces: End-to-End {SNR} and Achievable Rate},
	author={K. Ntontin and others},
	journal={arXiv:1908.07967v2},
	year={2019}
}

@ARTICLE{cute,
  author={Kim, Hongseok and others},
  journal=IEEE_J_WCOM, 
  title={A cross-layer approach to energy efficiency for adaptive {MIMO} systems exploiting spare capacity}, 
  year={2009},
  volume={8},
  number={8},
  pages={4264-4275},
  month=aug,
  }

@ARTICLE{yhFD22,
  author={Kim, Yonghwi and others},
  journal=IEEE_J_WCOM, 
  title={Low Complexity Frequency Domain Nonlinear Self-Interference Cancellation for Flexible Duplex}, 
  year={2025},
  volume={24},
  number={8},
  pages={6627-6642},
  month=aug,
  }

@ARTICLE{uavmec,
  author={Xiao, Han and others},
  journal=IEEE_J_WCOM, 
  title={Energy-Efficient {STAR-RIS} Enhanced {UAV}-Enabled {MEC} Networks With Bi-Directional Task Offloading}, 
  year={2025},
  volume={24},
  number={4},
  pages={3258-3272},
  month=apr,
  }

@ARTICLE{aerialmm,
  author={Xiong, Baiping and others},
  journal={IEEE Wireless Commun. Lett.}, 
  title={Performance Analysis of Aerial {RIS} Auxiliary {mmWave} Mobile Communications With {UAV} Fluctuation}, 
  year={2024},
  volume={13},
  number={4},
  pages={1183-1187},
  month=apr,
  }

@ARTICLE{arisb1,
  author={Faramarzi, Sajad and others},
  journal={IEEE Internet Things J.},
  title={Meta Reinforcement Learning for Resource Allocation in Aerial Active-{RIS}-Assisted Networks With Rate-Splitting Multiple Access}, 
  year={2024},
  volume={11},
  number={15},
  pages={26366-26383},
  month=aug,
  }

@ARTICLE{arisb2,
  author={Zhao, Jingjing and others},
  journal={IEEE Internet Things J.}, 
  title={Aerial Active {STAR-RIS}-Aided {IoT} {NOMA} Networks}, 
  year={2025},
  volume={12},
  number={8},
  pages={9525-9538},
  month=apr,
  }

@ARTICLE{arisb3,
  author={Wang, Dawei and others},
  journal=IEEE_J_JSAC, 
  title={Active Aerial Reconfigurable Intelligent Surface Assisted Secure Communications: Integrating Sensing and Positioning}, 
  year={2024},
  volume={42},
  number={10},
  pages={2769-2785},
  month=oct,
  }

@ARTICLE{ularis1,
  author={Yu, Xiao and others},
  journal=IEEE_J_TVT, 
  title={Channel Estimation for Irregular Subarrayed {RIS}-Aided {mmWave} Communications}, 
  year={2025},
  volume={74},
  number={11},
  pages={17247-17264},
  month=nov,
}

@ARTICLE{ularis2,
  author={Chen, Shengyao and others},
  journal=IEEE_J_TVT, 
  title={Interference Suppression for Active {RIS}-Empowered Array Radar Using Joint Beamforming Design}, 
  year={2025},
  volume={74},
  number={4},
  pages={6222-6238},
  month=apr,
  }

@ARTICLE{traj3d,
  author={Guan, Xiaoyi and others},
  journal=IEEE_J_TVT, 
  title={{3D} Trajectory Optimization for Fixed-Wing {UAV} Communications With Full {UAV} Dynamics}, 
 year={2025},
  volume={74},
  number={10},
  pages={15401-15415},
  }

@ARTICLE{beamgc,
  author={Deng, Qian and others},
  journal={IEEE Trans. Green Commun. Netw.}, 
  title={Adaptive Beam Alignment and Optimization for {IRS}-Aided High-Speed {UAV} Communications}, 
  year={2023},
  volume={7},
  number={3},
  pages={1583-1595},
  month=oct,
  }

@ARTICLE{add1,
  author={Sang, Jian and others},
  journal={IEEE Wireless Commun.}, 
  title={Coverage Enhancement by Deploying {RIS} in {5G} Commercial Mobile Networks: Field Trials}, 
  year={2024},
  volume={31},
  number={1},
  pages={172-180},
  month=feb,
  }

@ARTICLE{add2,
  author={Sang, Jian and others},
  journal=IEEE_J_WCOM,
  title={Multi-Scenario Broadband Channel Measurement and Modeling for Sub-6 {GHz} {RIS}-Assisted Wireless Communication Systems}, 
  year={2024},
  volume={23},
  number={6},
  pages={6312-6329},
  month=jun,
  }

@ARTICLE{add3,
  author={Liu, Yuanwei and others},
  journal={IEEE Communications Surveys \& Tutorials}, 
  title={Reconfigurable Intelligent Surfaces: Principles and Opportunities}, 
  year={2021},
  volume={23},
  number={3},
  pages={1546-1577},
  month={Third Quarter},
  }

@ARTICLE{af1,
  author={Zappone, Alessio and others},
  journal=IEEE_J_TSP, 
  title={Energy Efficiency Optimization in Relay-Assisted {MIMO} Systems With Perfect and Statistical {CSI}}, 
  year={2014},
  volume={62},
  number={2},
  pages={443-457},
  month=jan,
  }

@ARTICLE{af2,
  author={Sainath, B. and Mehta, Neelesh B.},
  journal=IEEE_J_WCOM, 
  title={Generalizing the Amplify-and-Forward Relay Gain Model: An Optimal {SEP} Perspective}, 
  year={2012},
  volume={11},
  number={11},
  pages={4118-4127},
  month=nov,
  }

@ARTICLE{af3,
  author={Chae, Chan-Byoung and others},
  journal=IEEE_J_TSP, 
  title={{MIMO} Relaying With Linear Processing for Multiuser Transmission in Fixed Relay Networks}, 
  year={2008},
  volume={56},
  number={2},
  pages={727-738},
  month=feb,
  }

@ARTICLE{ula1,
  author={Lei, Jiayi and others},
  journal={IEEE Trans. Commun.}, 
  title={{NOMA} for {STAR-RIS} Assisted {UAV} Networks}, 
  year={2024},
  volume={72},
  number={3},
  pages={1732-1745},
  month=mar,
  }

@ARTICLE{ula2,
  author={Gao, Ning and ohters},
  journal={IEEE Wireless Commun. Lett.}, 
  title={Aerial {RIS}-Assisted High Altitude Platform Communications}, 
  year={2021},
  volume={10},
  number={10},
  pages={2096-2100},
  month=oct,
  }

@ARTICLE{arisleo,
  author={Toka, Mesut and others},
  journal={IEEE Commun. Mag.}, 
  title={{RIS}-Empowered {LEO} Satellite Networks for {6G}: Promising Usage Scenarios and Future Directions}, 
  year={2024},
  volume={62},
  number={11},
  pages={128-135},
  month=nov,
  }

\end{document}